\documentclass[12pt]{article}
\usepackage{amsfonts,amsmath,amssymb,amsthm,graphicx}
\usepackage{fullpage}
\usepackage{changepage}
\usepackage{enumerate}
\usepackage{algorithm}
\usepackage{algorithmic}
\usepackage{scalerel}
\usepackage{accents}
\usepackage[margin=1in]{geometry}

\usepackage{todonotes}

\usepackage{natbib}
\usepackage{csquotes}
\usepackage{comment}

\usepackage{hyperref}
\usepackage[nameinlink]{cleveref}

\usepackage{tikz}
\usepackage{pgfplots}
\usetikzlibrary{patterns}
\usepgfplotslibrary{fillbetween}
\usetikzlibrary{intersections}

\usetikzlibrary{arrows, decorations.markings}

\tikzstyle{vecArrow} = [thick, decoration={markings,mark=at position
   1 with {\arrow[semithick]{open triangle 60}}},
   double distance=1.4pt, shorten >= 5.5pt,
   preaction = {decorate},
   postaction = {draw,line width=1.4pt, white,shorten >= 4.5pt}]
\tikzstyle{innerWhite} = [semithick, white,line width=1.4pt, shorten >= 4.5pt]

\theoremstyle{plain}
\newtheorem{theorem}{Theorem}[section]
\newtheorem{lemma}[theorem]{Lemma}
\newtheorem{claim}[theorem]{Claim}
\newtheorem{fact}[theorem]{Fact}

\newtheorem{corollary}[theorem]{Corollary}% reset theorem numbering for each chapter

\newenvironment{numberedtheorem}[1]{%
\renewcommand{\thetheorem}{#1}%
\begin{theorem}}{\end{theorem}\addtocounter{theorem}{-1}}

\newenvironment{numberedlemma}[1]{%
\renewcommand{\thetheorem}{#1}%
\begin{lemma}}{\end{lemma}\addtocounter{theorem}{-1}}

\theoremstyle{plain}
\newtheorem{definition}{Definition}[section] % definition numbers are dependent on theorem numbers
\newtheorem{example}[definition]{Example}
\newtheorem{assumption}[definition]{Assumption} 
\newtheorem{remark}[definition]{Remark}

\usepackage{amssymb}
\usepackage{amsmath}
\usepackage{ifthen}
\usepackage{fixmath}
\usepackage{xfrac}
\usepackage{sidecap}
\usepackage{subfig}
\usepackage{caption}

\DeclareMathAlphabet{\mathpzc}{OT1}{pzc}{m}{it}

\newcommand{\agind}[1][i]{_{#1}}

\newcommand{\ironed}{\bar}
\newcommand{\constrained}{\hat}
\newcommand{\optconstrained}{\composed{\optimized}{\constrained}}
\newcommand{\optimized}[1]{#1\opt}
\newcommand{\differentiated}[1]{#1'}

\newcommand{\tagged}[2]{{#2}^{#1}}

\newcommand{\starred}[1]{#1^\star}
\newcommand{\primedarg}[1]{#1\primed}
\newcommand{\noaccents}[1]{#1}
\newcommand{\composed}[3]{#1{#2{#3}}}

\newcommand{\newagentvar}[3][\noaccents]{%
\expandafter\newcommand\expandafter{\csname #2\endcsname}{#1{#3}}%
\expandafter\newcommand\expandafter{\csname #2s\endcsname}{#1{\boldsymbol{#3}}}%
\expandafter\newcommand\expandafter{\csname #2smi\endcsname}[1][i]{#1{\boldsymbol{#3}}_{-##1}}%
\expandafter\newcommand\expandafter{\csname #2i\endcsname}[1][i]{#1{#3}\agind[##1]}%
\expandafter\newcommand\expandafter{\csname #2ith\endcsname}[1][i]{#1{#3}_{(##1)}}%
}

\newcommand{\newitemvar}[3][\noaccents]{%
\expandafter\newcommand\expandafter{\csname #2\endcsname}{#1{#3}}%
\expandafter\newcommand\expandafter{\csname #2s\endcsname}{#1{\boldsymbol{#3}}}%
\expandafter\newcommand\expandafter{\csname #2smj\endcsname}[1][j]{#1{\boldsymbol{#3}}_{-##1}}%
\expandafter\newcommand\expandafter{\csname #2j\endcsname}[1][j]{#1{#3}_{##1}}%
\expandafter\newcommand\expandafter{\csname #2jth\endcsname}[1][j]{#1{#3}_{(##1)}}%
}

\newagentvar{alloc}{x}

\newagentvar{falloc}{z}

\newagentvar{quant}{q}
\newagentvar{randomquant}{{\hat{\quant}}}
\newagentvar[\constrained]{exquant}{\quant}
\newagentvar[\constrained]{critquant}{\quant}  %remove?
\newagentvar[\optconstrained]{monoq}{\quant}
\newagentvar{Val}{\nu}
\newagentvar{toquant}{Q}

\newagentvar{qprice}{\price}
\newagentvar{qrev}{R}
\newagentvar[\ironed]{iqrev}{\qrev}

\newagentvar{qalloc}{y}
\newagentvar{cumalloc}{Y}
\newagentvar[\constrained]{cumcalloc}{\cumalloc}
\newagentvar[\ironed]{ialloc}{\qalloc}
\newagentvar[\ironed]{icumalloc}{\cumalloc}

\newcommand{\exposted}[1]{#1^{\text{\it EP}}}
\newagentvar[\composed{\exposted}{\constrained}]{excalloc}{\qalloc}
\newagentvar[\exposted]{exalloc}{\qalloc}
\newagentvar[\exposted]{extalloc}{\talloc}
\newagentvar[\exposted]{exfalloc}{\falloc}

\newagentvar{rev}{R}
\newagentvar[\differentiated]{marg}{\rev}
\newagentvar{rawrev}{P}
\newagentvar[\differentiated]{rawmarg}{\rawrev}

\newagentvar[\primedarg]{inducedrev}{\rev}

\newagentvar[\tilde]{pseudorev}{\rev}
\newagentvar[\tilde]{pseudorawrev}{\rawrev}

\newagentvar{typespace}{{\cal T}}
\newagentvar{typesubspace}{S}

\newagentvar{type}{t}
\newagentvar{othertype}{s}
\newagentvar{val}{v}
\newagentvar{hval}{\bar \val}
\newagentvar{hbudget}{\bar \wealth}
\newagentvar{lbudget}{\underaccent{\bar}{ \wealth}}
\newagentvar{lowestval}{0}
\newagentvar{cumval}{V}
\newagentvar{cumprice}{P}
\newagentvar{welcurve}{W}
\newagentvar{revcurve}{R}
\newagentvar{outcome}{w}
\newagentvar{outcomespace}{{\cal W}}

\newcommand{\served}[1]{#1^1}
\newcommand{\nonserved}[1]{#1^0}
\newcommand{\alloced}[1]{#1^{\alloc}}
\newcommand{\allocedi}[1]{#1^{\alloci}}
\newagentvar[\alloced]{xoutcome}{\outcome}
\newagentvar[\allocedi]{xioutcome}{\outcome}
\newagentvar[\served]{soutcome}{\outcome}

\newagentvar[\nonserved]{nsoutcome}{\outcome}

\newagentvar{price}{p}
\newagentvar{randomprice}{{\hat{\price}}}
\newagentvar{talloc}{\alloc}
\newagentvar[\tilde]{balloc}{\alloc}
\newagentvar[\tilde]{bprice}{\price}

\newagentvar{act}{a}
\newagentvar{bidspace}{A}
\newagentvar{actspace}{A}

\newagentvar[\constrained]{critval}{\val}
\newagentvar[\constrained]{crittype}{\type}
\newagentvar[\constrained]{critvirt}{\virt}
\newagentvar[\constrained]{reserve}{\val} % should this be used???
\newagentvar[\constrained]{bidreserve}{\bid} % should this be used???
\newagentvar[\optconstrained]{monop}{\val}
\newagentvar[\constrained]{monot}{\type}
\newagentvar[\optimized]{monorev}{\rev}

\newagentvar{rcalloc}{y}
\newagentvar[\optimized]{optrcalloc}{\rcalloc}
\newagentvar{biddist}{G}

\newagentvar[\constrained]{critbid}{\bid}
\newagentvar[\constrained]{cbid}{B}
\newagentvar[\primedarg]{wbid}{\bid}
\newagentvar{gfunc}{\vartheta}
\newagentvar{block}{C}

\newagentvar{util}{u}

\newagentvar{strat}{s}
\newagentvar{bid}{b}

\newagentvar{virt}{\phi}
\newagentvar{cumvirt}{\Phi}
\newagentvar{qvirt}{\phi}
\newagentvar[\ironed]{ivirt}{\virt}
\newagentvar[\ironed]{icumvirt}{\cumvirt}

\newagentvar[\tagged{\text{SD}}]{sdvirt}{\virt}
\newagentvar[\composed{\tagged{\text{SD}}}{\ironed}]{sdivirt}{\virt}
\newagentvar[\tagged{\text{MD}}]{mdvirt}{\virt}
\newagentvar[\composed{\tagged{\text{MD}}}{\ironed}]{mdivirt}{\virt}

\newagentvar{dist}{F}
\newagentvar{dens}{f}
\newagentvar{hazard}{h}
\newagentvar{cumhazard}{H}

\newagentvar[\ironed]{iprice}{\price}  %% what for?
\newagentvar[\ironed]{ival}{\val}  %% what for?

\newagentvar[\ironed]{icumval}{\cumval}
\newagentvar{ints}{{\cal I}}

\newagentvar{wal}{w}

\newitemvar{pos}{j}
\newitemvar{weight}{w}
\newitemvar[\differentiated]{mweight}{\weight}
\newitemvar{udtype}{\type}
\newitemvar[\constrained]{udcrittype}{\type}
\newitemvar{udalloc}{\alloc}
\newitemvar{udprice}{\price}
\newitemvar[\constrained]{udcalloc}{\qalloc}
\newitemvar[\constrained]{udcumcalloc}{\cumalloc}

\newagentvar{mech}{{\cal M}}
\newagentvar[\skew{5}{\hat}]{cmech}{{\cal M}}
\newagentvar{alg}{{\cal A}}

\newcommand{\Rev}[2][]{\text{\bf Rev}\ifthenelse{\not\equal{}{#1}}{_{#1}}{}\!\left[{\def\givenn{\middle|}#2}\right]}

\newcommand{\lagrange}{\lambda}

\newcommand{\wealth}{w}

\DeclareMathOperator{\OPT}{OPT}

\newagentvar{trans}{\sigma}

\newagentvar{demandset}{S}

\newcommand{\opt}{^{\star}}
\newcommand{\primed}{^\dagger}
\newcommand{\doubleprimed}{^{\ddagger}}

\newagentvar{distout}{w}
\newagentvar{idistout}{\bar{w}}

\newagentvar{gap}{\delta}

\newagentvar[\starred]{calloc}{\alloc}

\newcommand{\R}{\mathbb R}
\newcommand{\APprivate}
{\rho\sqrt{2(\alphaprivate)(\betaprivate)}}
\newcommand{\alphaprivate}
{\budgetQuantile+2}
\newcommand{\betaprivate}
{\budgetQuantile+1}
\newcommand{\linearapproxratio}{\gamma}

\newcommand{\exanteprobability}{\hat{q}}

\newcommand{\pricerev}{{\rm AP}}

\newcommand{\exanteprice}{\hat{p}}
\newcommand{\clcumprice}{\hat{\cumprice}_{\lagrange^*}}
\newcommand{\lcumprice}{\cumprice_{\lagrange^*}}
\newcommand{\cumpricelinear}{\cumprice^L}
\newcommand{\correspondingprice}{\price\primed}
\newcommand{\correspondingquant}{\quant\primed}
\newcommand{\EX}{{\rm EX}}
\newcommand{\EXS}{{\rm EX\primed}}
\newcommand{\EXL}{{\rm EX\doubleprimed}}
\newcommand{\APF}{\tau}
\newcommand{\APFS}{\tau\primed}
\newcommand{\APFL}{\tau\doubleprimed}
\newcommand{\RevAPFww}{\Rev[\wealth]{\APFL_\wealth}}
\newcommand{\RevAPFwEw}{\Rev[\wealth^*]{\APFL_\wealth}}
\newcommand{\RevPPw}{\Rev[\wealth]{\marketclearing}}
\newcommand{\RevPPEw}{\Rev[\wealth^*]{\marketclearing}}

\newcommand{\instance}{{\cal I}}
\newcommand{\agents}{N}
\newcommand{\distribution}{{\cal D}}
\newcommand{\agentsnum}{n}

\newcommand{\exanterev}{{\rm EAR}}

\newcommand{\priceToQuantile}{Q}
\newcommand{\budgetQuantile}{\kappa}

\newcommand{\capacity}{C}
\newcommand{\mincapacity}{\underaccent{\bar}{C}}
\newcommand{\maxval}{\hval}

\newcommand{\myersonReserve}{m^*}

\newcommand{\valfunc}{V}
\newcommand{\optquant}{\hat{\quant}}

\newcommand{\distributions}{\boldsymbol{\cal D}}

\newcommand{\mhrbound}{3}

\newcommand{\improvedAP}
{\max\{\alpha, \sqrt{\alpha\beta \eta}\}}

\newcommand{\threshold}{\varrho}

\newcommand{\marketclearing}{\price^\quant}

\DeclareMathOperator{\argmax}{argmax}

\newcommand{\given}{\,\mid\,}

\newcommand{\prob}[2][]{\text{\bf Pr}\ifthenelse{\not\equal{}{#1}}{_{#1}}{}\!\left[{\def\givenn{\middle|}#2}\right]}
\newcommand{\expect}[2][]{\text{\bf E}\ifthenelse{\not\equal{}{#1}}{_{#1}}{}\!\left[{\def\givenn{\middle|}#2}\right]}

\newcommand{\tparen}{\big}
\newcommand{\tprob}[2][]{\text{\bf Pr}\ifthenelse{\not\equal{}{#1}}{_{#1}}{}\tparen[{\def\given{\tparen|}#2}\tparen]}
\newcommand{\texpect}[2][]{\text{\bf E}\ifthenelse{\not\equal{}{#1}}{_{#1}}{}\tparen[{\def\given{\tparen|}#2}\tparen]}

\newcommand{\sprob}[2][]{\text{\bf Pr}\ifthenelse{\not\equal{}{#1}}{_{#1}}{}[#2]}
\newcommand{\sexpect}[2][]{\text{\bf E}\ifthenelse{\not\equal{}{#1}}{_{#1}}{}[#2]}

\let\oldparagraph\paragraph
\renewcommand{\paragraph}[1]{\oldparagraph{#1.}}

\begin{document}

\pagenumbering{gobble}

%\input{Paper/cover.tex}
%\newpage
%\begin{titlepage}

\pagenumbering{arabic}
\setcounter{page}{1}
\title{Optimal Auctions vs.\ Anonymous Pricing: \\
Beyond Linear Utility}
\author{Yiding Feng\thanks{Department of Computer Science, Northwestern University. Email: \texttt{yidingfeng2021@u.northwestern.edu}.}
\and Jason D. Hartline\thanks{Department of Computer Science, Northwestern University.
Email: \texttt{hartline@northwestern.edu}.}
\and Yingkai Li\thanks{Department of Computer Science, Northwestern University.
Email: \texttt{yingkai.li@u.northwestern.edu}.} }
%\author{Submission 201}
\date{}

\maketitle

\begin{abstract}
    % For selling a single
    % item to 
    % agents with independent 
    % but non-identically 
    % distributed types,
    % the revenue optimal mechanism
    % is complex in 
    % linear utility settings,
    % and has
    % no closed form characterization 
    % for more non-linear settings.
    %  With respect to it,
    %  \citet{AHNPY-18} showed 
    %  that posting an anonymous price
    %  is an $e$-approximation 
    %  to the optimal
    %  revenue of the ex ante relaxation 
    %  of the auction problem
    %  for agents with linear utility and
    %  regular valuation distribution. 
    %  For agents with non-linear utility,
    %  we generalize the regularity concept and
    %  provide a reduction framework
    %  to analyze the approximation bound
    %  for anonymous pricing.
     
    %  Applying our reduction framework, 
    %  we derive constant approximation 
    %  bounds for public-budget utility,
    %  private-budget utility,
    %  and risk-averse utility.
    The revenue optimal mechanism for selling a single item to agents 
    with independent but non-identically distributed values is 
    complex for agents with linear utility 
    (\citealp{mye-81}) 
    and has no closed-form characterization 
    for agents with non-linear utility 
    \citep[cf.][]{AFHHM-12}. 
    Nonetheless, for linear utility agents
    satisfying a natural 
    regularity property, \citet{AHNPY-18} showed that simply posting 
    an anonymous price is an $e$-approximation.  
    We give a parameterization of the 
    regularity property that extends 
    to agents with non-linear utility 
    and show that the approximation 
    bound of anonymous pricing 
    for regular agents approximately
    extends to agents that satisfy this 
    approximate regularity property.  
    We apply this approximation 
    framework to prove that anonymous pricing 
    is a constant approximation to the 
    revenue optimal single-item auction for 
    agents with public-budget utility, 
    private-budget utility, 
    and (a special case of) risk-averse utility.
\end{abstract}

%\end{titlepage}

%\part{Paper: Optimal auctions vs. anonymous pricing: 
%beyond linear utility}

\setcounter{page}{1}

\section{Introduction}

In Bayesian mechanism design, 
a central question  
``\emph{simple versus optimal}''
%initiated by
%\citet{HR-09}
focuses on how well simple
mechanisms can approximate 
the revenue of the optimal mechanism
in complex environments.
The main motivation 
comes from the fact that 
the optimal mechanisms 
for asymmetric agents are usually 
both
difficult to derive and implement. 
For agents with linear utility, 
\citet{mye-81} characterized the optimal 
mechanism; it is sophisticated 
and involves discrimination.
For more general models 
(e.g.,
budgeted agents,
risk-averse agents),
either 
no closed form characterization 
is known in the literature,
or the optimal mechanisms
take the model too literally 
and are both fragile and impractical.
On the other hand, 
some simple mechanisms (e.g.,
posting an anonymous price,
second-price auction with anonymous reserve) 
are prevalent broadly. 

\citet{AHNPY-18} study 
%the approximation bound for 
the simplest 
mechanism for asymmetric agents with linear utility, 
namely, an anonymous pricing 
where an anonymous price is posted 
for selling a single good,
and the first agent (in an arbitrary order)
who values the good at higher 
price will buy the good. 
They upper bound the optimal revenue by 
considering the ex ante relaxation 
which sell at most 
one item in expectation over randomness of all
agents' values, i.e.,
the ex post feasibility constraint of selling at most 
one item is relaxed to that of selling 
at most one item 
ex ante.
% One of the main techniques in their work is the 
% revenue curve \citep[cf.][]{BR-89},
% which simplifies the characterization
% of the optimal mechanism 
% and evaluates approximation mechanisms for agents with linear utility.
They derive a tight $e$-approximation bound for 
the anonymous pricing to the ex ante relaxation 
for independent but non-identical 
agents with linear utility and 
{\em regular} valuation distributions (see below).
%\footnote{Regularity is a standard 
%assumption in mechanism design,
%and we defer its definition
%to the later discussion.}
%in \emph{Main Results}.
Their result implies that, 
up to an $e$ factor, 
discrimination and simultaneity 
are unimportant for optimizing revenue in 
single-item auctions.
A natural question  motivating our work
is
\begin{displayquote}
\textsl{
\hspace{-8pt}
Does the approximate optimality 
of anonymous pricing generalize 
to agents with non-linear 
utility under a suitable 
generalization of the 
regularity assumption?
}
\end{displayquote}

Regularity is a common assumption in mechanism design that simplifies
the derivation of the optimal mechanism \citep{mye-81} and enables
approximation mechanisms for agents with linear utility
\citep[e.g.,][]{HR-09}.  Fixing any class of mechanisms and a single
agent, the revenue curve is a mapping from a constraint $\quant$ on
the ex ante probability of sale, over randomness in the agent's type
and the mechanism, to the revenue of the optimal mechanism with the ex
ante constraint.  Specifically, the \emph{price-posting revenue curve}
is generated by fixing mechanism class to all price-posting
mechanisms; and the \emph{ex ante revenue curve} is by fixing
mechanism class to all possible mechanisms.  The regularity for linear
utility is defined as the equivalence of the price-posting revenue
curve and the ex ante revenue curve \citep[cf.][]{BR-89}.  These two
revenue curves are sufficient to pin down the revenue from anonymous
pricing, and ex ante relaxation, respectively.

Following the perspective of \citet{AFHH-13}, the methods of our paper
can be viewed as reductions in the following two senses.  First, we
approximately reduce the analysis of revenue of anonymous pricings for
non-linear-utility agents to the analysis of revenue of anonymous
pricings for linear-utility agents.  Thus, relative to anonymous
pricings, non-linear agent models can be considered well approximated
by linear agent models.  Second, our analysis reduces the multi-agent
question of approximation by an anonymous price to a collection of
single-agent approximation questions.  These single-agent
approximation questions ask whether or not the price-posting revenue
curve is a good approximation to the ex ante revenue curve.  Relative
to \citet{AFHH-13}, the single-agent problem to which we reduce gives
simpler mechanisms.

% In this paper,
% we introduce a generalization 
% of the regularity to the non-linear utility model, 
% give a general framework
% that 
% approximately 
% reduces the approximation bound 
% for anonymous pricing
% in
% some classical 
% non-linear utility models
% to the approximation bound for linear utility.
\paragraph{Main Results}
We introduce a generalization of regularity
that characterizes the gap between the price-posting revenue
curve and the ex ante revenue curve.
% The original regularity in linear utility model
% is the equivalence between
% the price-posting revenue curve 
% and the ex ante revenue curve 
% (i.e., no gap).
% These two revenue curves
% are sufficient to 
% pin down the revenue from the anonymous pricing,
% and the ex ante relaxation, respectively.
%
Based on this generalization,
we give a reduction framework to 
approximately reduce the analysis of the approximation bound
for anonymous pricing
for agents 
with non-linear utility  
to that of agents 
with linear utility.
% This framework includes two
% different versions:
% the first version in 
% \Cref{thm:general AP bound} requires the 
% generalized regularity only,
% while 
% the second version in
% \Cref{thm:general EAR bound} requires 
% the concavity of the price-posting revenue curve
% but on the other hand, a weaker generalized regularity 
% assumption. See \Cref{f:framework} for a graphical 
% comparison.
%
As the 
instantiations of the framework, we analyze 
the approximation bound for 
the anonymous pricing for 
asymmetric agents with 
public-budget utility,
private-budget utility,
and (a special case of) risk-averse utility, respectively. 

\emph{Public-budget Utility:} 
The first classical non-linear utility model
we consider is agents with 
public but non-identical
budgets.
% \citet{LR-96}
% characterize the optimal mechanism 
% for symmetric agents with values 
% drawn from the same 
% distribution which is 
% regular with decreasing density 
% (i.e., public-budget regular).
%regular and has decreasing density.
% \citet{CMM-11} characterize 
% the optimal mechanism for a single agent 
% with an arbitrary distribution.
For asymmetric agents with arbitrary distributions, the optimal
mechanism can be solved by a polynomial-time solvable linear program
over interim allocation rules \citep[cf.][]{AFHHM-12, CKM-13}, but no
closed-form characterization is known.  With the characterization of
the optimal mechanism under ex ante constraints in \citet{AFHH-13},
and a generalization of an argument from \citet{abr-06}, our framework
yields %$2e$-approximation and
an $e$-approximation bound, assuming the
valuation distributions are regular,
% and regular with decreasing 
% density,
% respectively 
(\Cref{thm:public e bound imporved}).

\emph{Private-budget Utility:} Mechanism design for agents with
private-budget utility is challenging because the agent types are
multi-dimensional.
% Using a generalized border's characterization, 
% the optimal mechanism can 
% be still solved via a convex program
% \citep[cf.][]{AFHHM-12}. 
A lot of work has focused on 
characterizing the optimal mechanism 
and approximation mechanisms 
in this setting (discussed subsequently). 
Our framework shows that with independent value and budget
distributions, regular value distributions, and with some assumptions
on the budget distribution that anonymous pricing is a constant
approximation to the optimal revenue (\Cref{thm:private bound},
\Cref{thm:private budget MHR}), e.g., for monotone hazard rate budget
distributions anonymous pricing is a $3e$ approximation.
% No result of optimal 
% or approximation mechanisms
% for asymmetric agents 
% has been given 
% in the previous literature.
% \footnote{\citet{CMM-11}
% consider a subclass of mechanism where
% the loser pays zero, and 
% provide a constant approximation within this 
% subclass.
% In the worst case, 
% the optimal mechanism in this subclass 
% is an $n$-approximation to the optimal
% mechanism,
% where $n$ is the number of agents.
% Additionally, they assume that the valuation 
% distribution is MHR, which is a stronger assumption
% than regularity.
% }

\emph{Risk-averse utility:}
It is standard to model risk-averse utility 
as a concave function that 
maps agents' wealth to a utility.
This introduces a non-linearity 
into the incentive constraints of the agents 
which in most cases makes mechanism design 
analytically intractable.
Most results for agents with 
risk-averse utility 
 consider the 
 comparative performance 
 of the first- and second-price auctions,
 cf.,
 \citet{RS-81},
 \citet{Hol-80},
 \citet{MR-84}.
\citet{mat-83}
and \citet{MR-84},
however,
characterize optimal mechanisms for
symmetric agents
for constant absolute risk aversion
and more general risk-averse model.
In this paper, we restrict attention to
very specific form of risk aversion studied in 
\citet{FHH-13},
which is called capacitated utility.
We derive 
a constant approximation bound
(\Cref{thm:risk-averse bound})
for anonymous pricing for asymmetric agents,
under an assumption on the capacities.

\begin{table}[t]
\begin{center}
\begin{tabular}{|c|c|c|c|c|}
    \hline
    & %\multicolumn{2}{c|}{
    public budget
    %} 
    & \multicolumn{2}{c|}{independent private budget} 
    & risk averse \\
    \hline
    % & \begin{minipage}{0.12\textwidth}
    % \vspace{3pt}
    % regular, decreasing density
    % \vspace{3pt}
    % \end{minipage} 
    & regular 
    & \begin{minipage}{0.15\textwidth}
    \vspace{3pt}
    value regular, budget MHR
    \vspace{3pt}
    \end{minipage} 
    & \begin{minipage}{0.18\textwidth}
    \vspace{3pt}
    value regular, \\
    budget exceeds \\expectation
    w.p.\\ 
    at least
    $\sfrac{1}{\budgetQuantile}$
    %expected budget has \\ quantile at least $\sfrac{1}{\budgetQuantile}$
    \vspace{3pt}
    \end{minipage}
    & \begin{minipage}{0.25\textwidth}
    \vspace{3pt}
    regular,
    support $[0, \maxval_i]$, \\
    capacity at least 
    $\sfrac{\maxval_i}{\eta}$
    \vspace{3pt}
    \end{minipage}\\
    \hline
    \begin{minipage}{0.07\textwidth}
    \vspace{3pt}
    approx ratio 
    \vspace{3pt}
    \end{minipage}
    & $e$ 
    % & $2e$ 
    & $\mhrbound e$ 
    & 
%     $\tfrac{(\budgetQuantile^2+2\budgetQuantile-1)
% (\budgetQuantile^2+\budgetQuantile-1
% )}
% {\budgetQuantile^2-1}e$ 
    $\sqrt{2(2+\budgetQuantile)
    (1+\budgetQuantile)}e$
    & $(2 + \ln \eta
    %\frac{\maxval}{\capacity}
    )e$ \\ 
    \hline
\end{tabular}
\vspace{-10pt}
\end{center}
\caption{Approximation bounds 
for anonymous pricing with asymmetric agents}
\vspace{-10pt}
\label{table:summary}
\end{table}
%\vspace{-15pt}

All the approximation bounds
and the corresponding assumptions
can be found in \Cref{table:summary}.
In each section, 
we also provide examples showing that 
without the assumptions 
we make, 
the constant approximation 
for anonymous pricing cannot be guaranteed.

% Beyond single-item environments, which are the main focus of the
% paper, our generalization of regularity can be plugged into other
% existing results for linear utility agents to approximately extend
% them to non-linear agents.  An example of sequential posted pricing
% for matroid environments is given in \Cref{sec:conclusion}.

\paragraph{Related Work}
A  prominent line of research has studied anonymous pricing in single-item environment for agents with linear utility.  
% \citet{DFK-16} characterize the revenue gap 
% between the anonymous pricing and 
% optimal pricing for multiple 
% independent and identically distributed 
% unit-supply agents 
% in the single item multi-unit setting. 
\citet{HR-09} show that 
second-price auction with an anonymous reserve is 
a $4$-approximation to the optimal revenue. 
\citet{AHNPY-18} improve this result by 
showing that anonymous pricing is a 
tight $e$-approximation 
to the optimal ex ante relaxation. 
\citet{JLTX-19} prove that the tight 
ratio between anonymous pricing 
and the
optimal
(discriminatory) 
sequential post pricing 
is $2.62$, 
and 
%\citet{GHKKKM-05, HH-15, DFK-16, JLTX-19}.
\citet{JLQTX-19} show that the same tight ratio holds 
between anonymous pricing and the optimal mechanism. 
% In the multi-item setting, 
% \citet{HH-15} provide sufficient condition 
% in the single agent setting for pricing 
% to be optimal for unit-supply utility and 
% additive utility. 

% A central technique used in our paper 
% is the application of the revenue curves. 
% \citet{BR-89} show that in the linear utility case, 
% the Myerson's auction can be interpreted as 
% the maximization of the marginal revenue, 
% i.e., the derivative of the revenue curve. 
% \citet{AFHH-13} extend this idea to multi-parameter agents
% and show that when the agents are 
% revenue linear, 
% %optimal auction also 
% maximizing the marginal revenue 
% is optimal. 
% Moreover, \citet{ala-11}, \citet{AFHH-13} results gave approximation 
% bounds for general environments
% (the marginal revenue mechanism), 
% but their mechanism for single item environments 
% is asymmetric, 
% and is more complex than the pricing mechanism. 
% Our paper continuous their development, 
% and our framework overcomes 
% these challenges, 
% which derives
% constant approximation bounds 
% for anonymous pricing mechanisms. 

One of the main contributions in our paper is to show that anonymous
pricing is a constant approximation for asymmetric agents with private
budgets.  There are several papers in the literature that consider
similar anonymous-pricing problems but with no unit-demand constraint on
agents.  \citet{abr-06} shows that market clearing (the anonymous
pricing where demand meets supply) gives a two approximation to the revenue
of the optimal mechanism for selling multiple units to a set of
asymmetric agents with public values and public budgets.
\citet{Richter-18} show that under the assumption that the valuation
distribution is regular with decreasing density, a price-posting
mechanism is optimal for selling a divisible good to a continuum of
agents with private budgets.  The lack of a unit-demand constraint on
agents is crucial for their analyses and its introduction poses
significant challenges for ours.
% In fact, as shown in \citet{CG-00}, 
% the optimal mechanism 
% for a single agent with unit-supply constraint 
% is to post a complicated
% convex pricing function. 
% For multiple agents with unit-supply constraint 
% and a private budget, 
% \citet{PV-14} provide a characterization for 
% the optimal mechanism 
% when the budget distribution is uniform. 

%For agents with unit-supply constraint, 

In the single item environment, for a single agent with private budget
constraint, when her valuation distribution satisfies declining
marginal revenues, \footnote{The decreasing marginal revenue
  assumption considers revenue as a function of price and requires its
  derivative be decreasing.}  \citet{CG-00} characterize the optimal
mechanism by a differential equation.  \citet{DW-17} characterize the
optimal mechanism for a single agent with an arbitrary distribution by
a linear program and use an algorithmic approach to construct the
solution.  For multi-agents settings, \citet{LR-96} show that the
all-pay auction is optimal for symmetric agents in the public budget
setting when the valuation distribution is regular and decreasing
density.  \citet{PV-14} generalize the characterization of the optimal
mechanism for symmetric agents in the private budget setting, with
budgets distributed uniformly.

For more general feasibility constraints and private budgets,
\citet{CMM-11} show that when the feasible allocations form a matroid
and the valuation distribution is monotone hazard rate (MHR), a simple
lottery mechanism is a constant approximation to the optimal pointwise
individually rational mechanism.
\footnote{ The pointwise individually rational mechanism requires that
  the payment of the agent is at most her value of the allocation
  after the realization of the randomness of the mechanism.}  Note
that, via an example in \citet{CMM-11}, there can be a linear gap in
revenue between the optimal pointwise individually rational mechanism
and the optimal interim individually rational mechanism. Under the more
classical interim individual rationality constraint, \citet{AFHHM-12}
reduce the multi-agent problem to the single-agent interim
optimization problem, and solve it via a convex program.
\citet{AFHH-13} approximately reduce the multi-agent problem to the
single-agent ex ante optimization problem, and show, for example, that
in the special case of the single item environment, (discriminatorally
and) sequentially posting the single agent ex ante optimal mechanisms
gives an $\sfrac{e}{(e-1)}$-approximation.  In comparison with these
reduction results, the advantage of our approach is that while the ex
ante optimal mechanisms are used in our bound, the approximately
optimal mechanisms we identify are based only on single-agent posted
pricing and are, thus, much simpler.  Moreover, our resulting simple
mechanisms can be optimized over directly.

\section{Preliminaries}
\label{sec:prelim}

We consider the single
item environment 
from auction theory.
A seller has a single indivisible good.
Agents have private types 
drawn independently,
but not necessarily identically from 
some distributions.
We are interested in optimizing the revenue, 
namely the sum of payments made by agents, 
of a mechanism for selling the good. 

\paragraph{Auction Instance} 
An auction instance is defined as $\instance = (\agents, \typespaces,
\distributions, \utils)$.  Here $\agents$ is the set of agents and
$\agentsnum = |\agents|$ is the number of agents.  $\typespaces =
\times_{i \in \agents} \typespace_i, \distributions = \times_{i \in
  \agents} \distribution_i$, and 
$\utils = \times_{i \in \agents}
\util_i$ are the type space, distributions and utility functions for
each agent.  In this paper, agents are non-identical, and the outcome
for an agent $i$ is the distribution over the pair $(\alloci,
\pricei)$, where allocation $\alloci \in \{0,1\}$ and payment $\pricei
\in \R_+$.  The utility function of each player $\util_i$ is a mapping
from her private type and the outcome to her von Neumann-Morgenstern
utility for the outcome.

\begin{comment}
\paragraph{Mechanisms} 
In this paper, we consider mechanisms that are Bayesian incentive
compatible (BIC) and interim individual rational (IIR).  The set of
all BIC-IIR mechanisms
%which maps types to allocations and payments
under the type structure and utility function 
is convex. 
Moreover, 
we mainly focus on two special mechanisms, 
the revenue optimal mechanisms 
and the anonymous pricing mechanisms. 
%We use $\OPT(\instance)$ to denote the optimal revenue among all truthful mechanisms for instance $\instance$, 
%and we use $\pricerev(\instance, \price)$ to denote the revenue of posting unit price $\price$ to all agents for instance $\instance$. 
\end{comment}

\paragraph{Type structures and utility functions}
The general framework of the paper will be instantiated in four
specific models.
\begin{enumerate}
    \item \textbf{Linear utility:} 
    A private type $\type$ is a private value
    $\val \in [0,\hval]$
    of the good. 
    We denote the cumulative distribution function 
    and the density function
    for the valuation distribution
    by $F$ and $f$ respectively.
    Given allocation $\alloc$ and 
    payment $\price$,
    an agent's utility is $\val\alloc - \price$.
    Our development of the approximation bound 
    for anonymous pricing is based on 
    a reduction from 
    general utility function environments to 
    linear utility environments.
    \item \textbf{Public-budget utility:}
    A private type $\type$ is  a 
    private value $\val\in [0,\hval]$.
    The utility function $\util$ also encodes
    a public budget $\wealth\in \R_+$,
    which is not necessarily identical 
    across agents.
    Given allocation $\alloc$ and 
    payment $\price$,
    an agent has utility $\val\alloc - \price$
    if the payment $\price$ is at most the budget $\wealth$ and 
    utility $-\infty$ otherwise.
    \item \textbf{Private-budget utility:}
    A private type $\type$ is a pair $(\val, \wealth)$ 
    that consists of a 
    private value $\val\in [0,\hval]$ 
    and a private budget $\wealth\in[\lbudget, \hbudget]$.
    We denote the cumulative distribution function 
    and 
    the density function
    for the budget distribution 
    by $G$ and $g$ respectively.
    Given allocation $\alloc$ and 
    payment $\price$,
    her utility is $\val\alloc - \price$
    if the payment $\price$ is at most the budget $\wealth$ and 
    is $-\infty$ otherwise.
    \item \textbf{Risk-averse utility:}
    A private type $\type$ is 
    a private value $\val \in [0, \hval]$
    of the good.
    The utility function $\util$
    is a concave function mapping 
    from the wealth
    of an agent
    to a utility.
    Specifically, we restrict attention 
    to a very specific form of risk aversion
    studied in \citet{FHH-13},
    which is both computationally 
    and analytically tractable: 
    utility functions that are 
    linear up to a given capacity $\capacity$
    and then flat.
    Given allocation $\alloc$ and payment $\price$,
    an agent has utility 
    $\min\{\val\alloc - \price, \capacity\}$.
    We refer to this utility function
    as \emph{capacitated utility}. 
    The capacity $\capacity$ 
    is encoded in the utility 
    function and is not necessarily identical  
    across agents.
\end{enumerate}
The assumption that the supports of the distributions are bounded
intervals in the models above is for technical simplicity and is
without loss of generality.  
% Where it is appropriate we denote
% simplified outcomes as the pair $(\alloc, \price)$, where $\alloc \in
% [0,1]$ is the probability the agent is allocated with the item, and
% $\price \in \R_+$ is the expected payment.

\paragraph{Ex ante revenue curves and price-posting revenue curves}

%\citet{BR-89} reduce multi-agent mechanism design problem 
%to that of solving a collection 
%of simple single-agent pricing problems under 
%the linear utility assumption. 

We introduce the \emph{quantile space}, \emph{ex ante revenue curves}
and \emph{price-posting revenue curves}.

\begin{definition}
The \emph{quantile} $\quant$ of 
a single-dimensional agent 
with value $\val$ drawn from distribution $F$ 
is the measure with respect to $F$ of stronger values, 
i.e., $\quant = 1 - F(\val)$; 
the inverse demand function $\valfunc$ maps 
an agent's quantile to her value, 
i.e., $\valfunc(\quant) = F^{-1}(1 - \quant)$.
\end{definition}
\begin{definition}
Given ex ante constraint $\quant$, 
the single-agent \emph{ex ante revenue-maximization problem}
is to find the optimal mechanism 
with ex ante allocation 
probability 
(i.e.\ the expected allocation over the draws of the agent's type)
exactly $\quant$. 
The optimal ex ante revenue, 
as a function of $\quant$, 
is denoted by the \emph{ex ante revenue curve} $\revcurve(\quant)$.
\end{definition}

Since the revenue (i.e., expected payment) is a linear objective, and
the space of feasible mechanisms is convex, the ex ante revenue curve
is concave.
\begin{fact}\label{lem:ex-ante is concave}
The ex ante revenue curve is concave. 
\end{fact}

\begin{definition}
A \emph{per-unit price} $\price:[0,1]\rightarrow \R_+$ for selling an
indivisible good probabilistically is a mapping from a lottery with
winning probability $\alloc\primed$ to a payment $\price\primed =
\price\cdot \alloc\primed$.  \emph{Per-unit pricing} for a single
agent is a mechanism where a per-unit price $\price$ is posted, and
the agent can select a lottery with an arbitrary winning probability
$\alloc\primed$ and pay $\price\cdot \alloc\primed$.  The item remains
unsold with probability $1-\alloc\primed$.
\end{definition}

\begin{definition}
Given ex ante constraint $\quant$, the single-agent \emph{ex ante
  price-posting problem} is to find the per-unit pricing with
ex ante allocation probability exactly $\quant$.  The optimal ex ante
price-posting revenue, as a function of $\quant$, is denoted by the
\emph{price-posting revenue curve} $\cumprice(\quant)$.  The
\emph{market clearing price $\marketclearing$} for the ex ante constraint
$\quant$ is $\marketclearing = \cumprice(\quant)/\quant$.
\end{definition}

For an agent with linear utility, the price-posting revenue curve
$\cumprice(\quant)$ at any quantile~$\quant$ is achieved by posting
per-unit price $V(\quant)$, i.e., $\cumprice(\quant) = \quant
\valfunc(\quant)$; and \citet{BR-89} give a geometric connection
between the ex ante revenue curve and price-posting revenue curve.
\begin{lemma}[\citealp{BR-89}]
\label{lem:price-posting revenue curve to ex ante revenue curve}
For an agent with linear utility, the ex ante revenue curve
$\revcurve$ is equal to the concave hull of the price-posting revenue
curve $\cumprice$.
\end{lemma}
For agents with linear utility, an agent's valuation distribution is
\emph{regular} if the price-posting revenue curve is concave;
equivalently, if the price posting and ex ante revenue curves are
identical, i.e., $\revcurve=\cumprice$.  The price that maximizes the
revenue from a single agent is called monopoly reserve
$\myersonReserve$.  
\Cref{lem:price-posting revenue curve to ex ante revenue curve} 
implies that the ex ante optimal mechanisms for a
single agent with linear utility, is either a price posting or the
randomization over two price postings.  For general utility agents,
the ex ante revenue curve can be much larger than the price-posting
revenue curve almost everywhere (\Cref{exp:MHR fail}).

It is straightforward to map per-unit price to quantile with the
price-posting revenue curve, i.e., $\priceToQuantile(\price,
\cumprice) = \argmax\{\quant:\cumprice(\quant)=\quant \price\}$.  In
this paper, it will be useful to imagine that an ex ante revenue curve
$\revcurve$ is generated by per-unit pricing, even when it is not, and
mathematically define the mapping from \emph{effective price} to
quantile in the same way, i.e., $\priceToQuantile(\price, \revcurve) =
\argmax\{\quant:\revcurve(\quant)=\quant \price\}$. Equivalently, this
definition imagines a regular linear-utility agent with the same
revenue curve $\revcurve$ and for which $\priceToQuantile(\price,
\revcurve)$ is the largest quantile that accepts price $\price$.

\begin{definition}
\emph{Effective price posting} $\price$ 
to an ex ante revenue curve $\revcurve$
is a mechanism 
where a per-unit price $\price$
is posted to an agent with a 
price-posting revenue curve equivalent to $\revcurve$.
\end{definition}

\paragraph{Different mechanisms of interest in this paper}
Through this paper, we focus on 
optimal mechanisms in
the following three classes.
\begin{enumerate}
    \item \textbf{Auctions:}
    An auction is any mechanism that maps types to allocations and payments
    subject to incentive and feasibility constraints. 
    Under the linear utility assumption,
    the optimal auction was characterized by
    \citet{mye-81} and this characterization, though complex, 
    is the foundation of modern auction theory.
    For more general utility models with asymmetric agents
    (e.g.\ public-budget utility,
    %\todo[inline]{I think there are nice characterizations of optimal auctions in the public-budget model.}, 
    private-budget utility, risk-averse utility),
    there is no closed form characterization in the literature.
    
    \item \textbf{Ex ante relaxations:} The ex ante relaxation
      considers the problem of selling at most one item in expectation
      over draws of all agents' types, i.e., the ex post feasibility
      constraint of selling at most one item is relaxed to that of
      selling at most one item ex ante.  Fixing the ex ante
      probability of serving each agent, the ex ante optimal mechanism
      solves the single-agent ex ante revenue-maximization problem for
      all agents individually.  The ex ante optimal mechanism for
      asymmetric agents usually discriminates and may, ex post,
      simultaneously serve multiple agents.  This relaxation was
      identified as a quantity of interest in \citet{CHK-07} and its
      study was refined by \citet{ala-11} and \citet{yan-11}.  The
      revenue of the ex ante optimal mechanism gives an upper bound on
      the revenue of the (point-wise feasible) optimal auction, see
      e.g., \citet{ala-11}.  We denote the optimal revenue achieved by
      the ex ante relaxation for a specific collection of ex ante
      revenue curves $\{\revcurve_i\}_{i=1}^n$ by
      $\exanterev(\{\revcurve_i\}_{i=1}^n)
      =\max\limits_{\quant_i:{\scriptstyle\sum}\quant_i \leq 1}\sum_i
      \revcurve_i(\quant_i)$.
    
    \item \textbf{Anonymous pricings:} An anonymous pricing mechanism
      posts an anonymous per-unit price $\price$ and the agents arrive
      in an arbitrary order.  Each agent can select a lottery with an
      arbitrary winning probability $\alloc\primed$ and pay
      $\price\cdot \alloc\primed$.  The lottery executes immediately
      before the next agent arrives.  Once an agent wins the item, the
      mechanism halts.  We denote the revenue achieved by posting
      anonymous per-unit price $\price$ for a specific collection of
      price-posting revenue curves $\{\cumprice_i\}_{i=1}^n$ by
      $\pricerev(\{\cumprice_i\}_{i=1}^n,\price)$.  We denote the
      revenue from the optimal anonymous price by
      $\pricerev(\{\cumprice_i\}_{i=1}^n) = \max_\price
      \pricerev(\{\cumprice_i\}_{i=1}^n,\price)$.
    
    Consider an arbitrary agent $i$.
    Suppose all agents who come earlier 
    than her
    do not win the item.
    The 
    probability that
    she wins the item is 
    $\priceToQuantile(\price,\cumprice_i)$,
    where the randomness comes from her own 
    type and the lottery which she selects.
    Thus,
    $\pricerev(\{\cumprice_i\}_{i=1}^n,\price)
    =\price\cdot\left(1 - \prod_i(1 - \priceToQuantile(\price,\cumprice_i))\right)
    $.
\end{enumerate}

\paragraph{Linear utility $\boldsymbol{e}$-approximation bound}
\citet{AHNPY-18} show that anonymous pricings give a tight
$e$-approximation to ex ante relaxations for agents with regular
valuation distributions and linear utility.  For regular linear
agents, the price-posting revenue curves are concave and equal to the
ex ante revenue curves.  Moreover, for any non-negative concave
function on domain $[0,1]$, there is a regular value distribution with
this function as its revenue curve.  Thus, the result of
\citet{AHNPY-18} can be restated as follows.

\begin{theorem}[\citealp{AHNPY-18}]\label{thm:linear bound}
Any set $\{\revcurve_i\}_{i=1}^n$ of non-negative concave functions on
domain $[0,1]$ satisfies $\rho \cdot \pricerev(\{\revcurve_i\}_{i=1}^n) \geq
\exanterev(\{\revcurve_i\}_{i=1}^n)$ for $\rho$ which
numerically evaluates to $e \approx 2.718$; i.e., anonymous pricing is
a $\rho$ approximation to the ex ante relaxation for linear utility
agents with values drawn from regular distributions with these revenue
curves.
\end{theorem}

\section{A General Reduction Framework}
\label{sec:framework}

In this section, we give two approaches that approximately
reduce the approximation of anonymous pricing for non-linear agents to
the approximation of anonymous pricing for linear agents.  The first
approach applies when a ``closeness'' condition between the
price-posting revenue curve and the ex ante revenue curve holds.  The
approximation guarantee degrades with this closeness condition.  The
second approach applies under a much weaker closeness condition but
additionally requires that the price-posting revenue curve is concave.
The high-level steps of these reductions are depicted in
\Cref{f:framework}.  

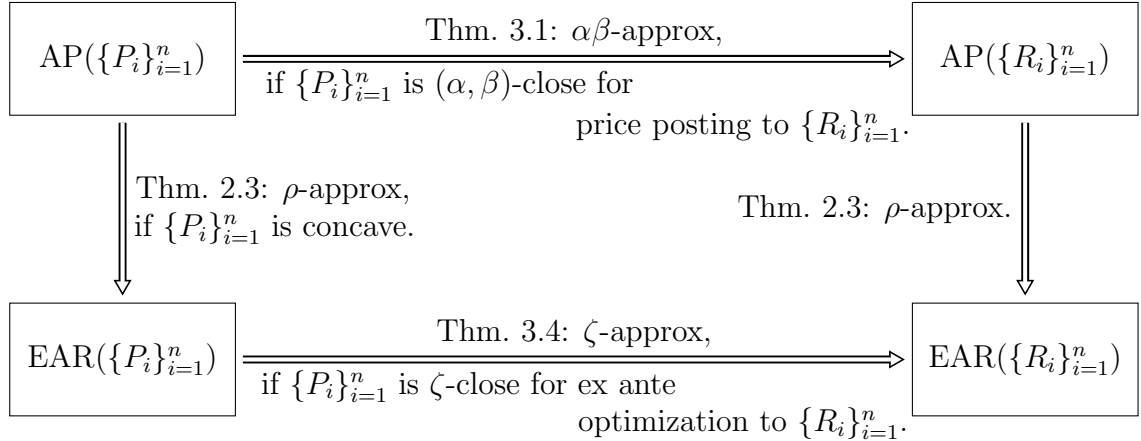
\begin{figure}[t]
\centering
\begin{tikzpicture}[scale = 0.5]

%short version

% \draw (0,0) rectangle (6,3) node[pos=.5] 
% {$\exanterev(\{\cumprice_i\}_{i=1}^n)$};

% \draw (0,8) rectangle (6,11) node[pos=.5]
% {$\pricerev(\{\cumprice_i\}_{i=1}^n)$};

% \draw (22,0) rectangle (28,3) node[pos=.5] 
% {$\exanterev(\{\revcurve_i\}_{i=1}^n)$};

% \draw (22,8) rectangle (28,11) node[pos=.5]
% {$\pricerev(\{\revcurve_i\}_{i=1}^n)$};

% \draw [vecArrow] (6.2, 1.5) -- (21.8, 1.5);
% \draw (14, 2.2) node {
% [Thm.~\ref{thm:general EAR bound}] 
% $\alpha\beta$-approx,};
% \draw (14, 0.8) node {if 
% $\{\cumprice_i\}%_{i=1}^n
% $ 
% is downward $(\alpha,\beta)$-close to
% $\{\revcurve_i\}%_{i=1}^n
% $.};

% \draw [vecArrow] (6.2, 9.5) -- (21.8, 9.5);
% \draw (14, 10.2) node {
% [Thm.~\ref{thm:general AP bound}]
% $\alpha\beta$-approx,};

% \draw (14, 8.8) node {if 
% $\{(\cumprice_i, \revcurve_i)\}$
% is $(\alpha, \beta)$-regular.};

% \draw [vecArrow]  (3, 7.8) -- (3, 3.2);
% \draw (7.0, 6.0) node {
% [Thm.~\ref{thm:linear bound}]
% $\rho$-approx,};
% \draw (6.6, 5.0) node {if $\{\cumprice_i\}$
% is concave.};

% \draw [vecArrow]  (25, 7.8) -- (25, 3.2);
% \draw (21.0, 6.0) node {
% [Thm.~\ref{thm:linear bound}]
% $\rho$-approx.};

% long version

\draw (0,0) rectangle (6,3) node[pos=.5] 
{$\exanterev(\{\cumprice_i\}_{i=1}^n)$};

\draw (0,8) rectangle (6,11) node[pos=.5]
{$\pricerev(\{\cumprice_i\}_{i=1}^n)$};

\draw (24,0) rectangle (30,3) node[pos=.5] 
{$\exanterev(\{\revcurve_i\}_{i=1}^n)$};

\draw (24,8) rectangle (30,11) node[pos=.5]
{$\pricerev(\{\revcurve_i\}_{i=1}^n)$};

\draw [vecArrow] (6.2, 1.5) -- (23.8, 1.5);
\draw (15, 2.2) node {
Thm.~\ref{thm:general EAR bound}: 
$\zeta$-approx,};
\draw (12.2, 0.8) node {if 
$\{\cumprice_i\}_{i=1}^n
$ 
is $\zeta$-close 
for ex ante };

\draw (19.5, -0.2) node{
optimization to
$\{\revcurve_i\}_{i=1}^n
$.};

\draw [vecArrow] (6.2, 9.5) -- (23.8, 9.5);
\draw (15, 10.2) node {
Thm.~\ref{thm:general AP bound}:
$\alpha\beta$-approx,};

\draw (11.7, 8.8) node {if 
$\{\cumprice_i\}_{i=1}^n
$ 
is $(\alpha,\beta)$-close 
for };

\draw (19.5, 7.7) node{
price posting to
$\{\revcurve_i\}_{i=1}^n
$.};

\draw [vecArrow]  (3, 7.8) -- (3, 3.2);
\draw (7.0, 6.0) node {
Thm.~\ref{thm:linear bound}:
$\rho$-approx,};
\draw (7.05, 5.0) node {if $\{\cumprice_i\}_{i=1}^n
$
is concave.};

\draw [vecArrow]  (27, 7.8) -- (27, 3.2);
\draw (23.0, 5.5) node {
Thm.~\ref{thm:linear bound}:
$\rho$-approx.};
\end{tikzpicture}
\caption{\label{f:framework} The reduction framework.  The upper path
  uses the approximate robustness of anonymous pricing. A better bound
  for the upper path is obtained from \Cref{thm:improved AP bound} but
  omitted from the figure.  The lower path uses the approximate
  robustness of the ex ante relaxation and requires that the
  price-posting revenue curves are concave.}
\end{figure}

These reductions are based on the facts that the ex ante revenue is
given mathematically by the ex ante revenue curves
$\{\revcurve_i\}_{i=1}^n$, the anonymous pricing revenue is given by
the price-posting revenue curves $\{\cumprice_i\}_{i=1}^n$, and the
$\rho$ approximation of anonymous pricing to the ex ante relaxation
for linear agents implies, as stated in \Cref{thm:linear bound}, an
approximation bound between anonymous pricing $\pricerev(\cdot)$ and
the ex ante relaxation $\exanterev(\cdot)$ for any set of concave
revenue curves.

% \begin{definition}
% \label{def:general regularity}
% An agent's price-posting revenue curve $\cumprice$
% \emph{$(\alpha, \beta)$-approximates}
% his ex ante revenue curve $\revcurve$ at an
% effective-price 
% $\exanteprice$, 
% if $\priceToQuantile(\frac{\exanteprice}{\alpha}, \cumprice) 
% \geq \frac{1}{\beta} 
% \priceToQuantile(\exanteprice, \revcurve)$.
% We call the agent 
% \emph{$(\alpha, \beta)$-regular}
% if this condition holds at all $\exanteprice$. 
% We define $(1,1)$-regular as regular. 
% \end{definition}

% It is instructive to see some examples  
% under the linear utility assumption.

% \begin{example}
% Given the valuation distribution
% is regular,
% agents with linear utility function 
% is $(1, 1)$-regular.
% The anonymous pricing is a
% $\rho$-approximation to 
% the ex ante relaxation.
% \end{example}

% \begin{example}
% Given the valuation distribution
% is arbitrary (possibly irregular),
% agents with linear utility function 
% is $(\infty, \infty)$-regular.
% The anonymous pricing is an
% $n$-approximation 
% to the ex ante relaxation
% where $n$ is the number of agents.
% \end{example}

\subsection{Approximate Robustness of  Anonymous Pricing}

We develop the upper path of \Cref{f:framework} for reducing
non-linear anonymous pricing problem to the linear anonymous pricing
problem.  This upper path is based on identifying a closeness property
on revenue curves that approximately preserves anonymous pricing
revenue.  This closeness property, illustrated in
the left hand side of \Cref{f:definition},
can be viewed as an approximate notion of regularity.

\begin{definition}
The agent's price-posting revenue curve $\cumprice$ 
is \emph{$(\alpha, \beta)$-close for price posting} 
to her ex ante revenue curve $\revcurve$, 
if 
$\cumprice(\quant) \geq \frac{1}{\alpha} \revcurve(\quant)$ 
for all
$\quant \in [0, 1/\beta]$.  Such an agent is \emph{$(\alpha, \beta)$-close for price posting}.
\end{definition}

\begin{theorem}\label{thm:general AP bound}
Given agents with 
ex ante revenue 
curves 
$\{\revcurve_i\}_{i=1}^n$
and price-posting revenue curves
$\{\cumprice_i\}_{i=1}^n$,
% if the revenue of posting effective price $\exanteprice$ 
% on revenue curve $\{\revcurve_i\}_{i=1}^n$
% is a 
% $\linearapproxratio$-approximation to the ex ante relaxation
% revenue of those agents, 
% i.e., 
% $\pricerev(
% \{\revcurve_i\}_{i=1}^n, \exanteprice) 
% \geq \frac{1}{\linearapproxratio} \exanterev(\{\revcurve_i\}_{i=1}^n)$, 
% and 
if each agent $i$, 
is $(\alpha, \beta)$-close for price posting, 
anonymous pricing on the price posting revenue curve  
is an $(\alpha\beta)$-approximation
to anonymous pricing on the ex ante revenue curve,
i.e.,
$\alpha \beta \cdot 
\pricerev(\{\cumprice_i\}_{i=1}^n)
\geq 
\pricerev(\{\revcurve_i\}_{i=1}^n)$. 
\end{theorem}
\begin{proof}
Let $\exanteprice$ be the anonymous price on ex ante revenue curve, 
and let the per-unit price 
$\correspondingprice = \frac{\exanteprice}{\alpha}$. 
For each agent $i$, 
let $\exanteprobability_i
= \priceToQuantile(\exanteprice, \revcurve_i)$
and  
$\correspondingquant_i
= \priceToQuantile(\frac{\exanteprice}{\alpha}, \cumprice_i)$. 
If the quantile $\exanteprobability_i \leq \frac{1}{\beta}$, 
since $\cumprice_i$ is $(\alpha, \beta)$-close for price posting to $\revcurve_i$, 
$\cumprice_i(\exanteprobability_i) \geq \frac{1}{\alpha} \revcurve_i(\exanteprobability_i)$. 
Therefore, 
$\frac{\exanteprice_i \cdot \exanteprobability_i}{\alpha}
\leq \cumprice_i(\exanteprobability_i)$, 
which implies $\correspondingquant_i \geq \exanteprobability_i$. 

If $\exanteprobability_i > \frac{1}{\beta}$, 
and if $\correspondingquant_i \geq \frac{1}{\beta}$, 
then $\correspondingquant_i \geq \frac{\exanteprobability_i}{\beta}$. 
If $\correspondingquant_i < \frac{1}{\beta}$,
$(\alpha, \beta)$-closeness
for price posting 
implies that
$\revcurve_i(\correspondingquant_i) 
\leq \alpha \cumprice_i(\correspondingquant_i)
=\alpha\cdot \frac{\exanteprice}{\alpha} \cdot \correspondingquant_i$,
namely,
$\exanteprice \cdot \correspondingquant_i %= \alpha \correspondingprice \cdot \correspondingquant 
\geq \revcurve_i(\correspondingquant_i) = 
\exanteprice \cdot \exanteprobability_i$, 
which contradicts to the assumption that 
$\exanteprobability_i > \frac{1}{\beta} \geq \correspondingquant_i$. 

% By our assumption, 
% the quantile of agent $i$ for posting price $\frac{\exanteprice}{\alpha}$ is 
% $\priceToQuantile(\frac{\exanteprice}{\alpha}, \cumprice_i) 
% \geq \frac{1}{\beta}
% \priceToQuantile(\exanteprice, \revcurve_i)$. 
Therefore, the revenue of posting price $\frac{\exanteprice}{\alpha}$ on the price posting revenue curve is 
\begin{eqnarray*}
    \pricerev\left(\{\cumprice_i\}_{i=1}^n,
    \frac{\exanteprice}{\alpha}\right) 
    &=& \frac{\exanteprice}{\alpha} \cdot 
    \left(1 - \prod\nolimits_i (1-\correspondingquant_i) \right) 
    \geq \frac{\exanteprice}{\alpha} \cdot 
    \left(1 - \prod\nolimits_i (1 - \frac{\exanteprobability_i}{\beta}) \right) \\
    &\geq& \frac{\exanteprice}{\alpha\beta} \cdot \left(1 - \prod\nolimits_i 
    (1 - \exanteprobability_i) \right) 
    = \frac{1}{\alpha\beta}
    \pricerev(\{\revcurve_i\}_{i=1}^n, \exanteprice). \qedhere
\end{eqnarray*}
\end{proof}

\begin{figure}[t]
\begin{flushleft}
\hspace{-5pt}
\begin{minipage}[t]{0.49\textwidth}
\centering
\setlength{\unitlength}{1cm}
\thinlines
\begin{tikzpicture}[scale = 0.5]
%\fill[gray!40!white] (0, 0) -- (8, 0) -- (8, 7) -- (0, 7);

\draw (-0.2,0) -- (12.5, 0);
\draw (0, -0.2) -- (0, 6);

\begin{scope}[very thick]
\draw plot [smooth, tension=0.6] coordinates { (0,0) (1,3.6) (5, 5) (12,0)};
\end{scope}

\draw plot [smooth, tension=0.6] coordinates { (0,0) (0.8,2) (4, 3.5) (7.5, 1.5) (9,1.1) (12,0)};

\draw [dotted] plot [smooth, tension=0.7] coordinates { (0,0) (1,1.8) (4, 2.6) (6, 2) };
\draw [dotted] (6,0) -- (6, 2);

\draw (0, -0.8) node {$0$};
\draw (6, -0.8) node {$\sfrac{1}{\beta}$};
\draw (12, -0.8) node {$1$};

\draw (3.6, 5.8) node {$\revcurve$};
\draw (3.6, 4.1) node {$\cumprice$};
\draw (3.6, 1.8) node {$\sfrac{\revcurve}{\alpha}$};
% \draw (3.4, -0.8) node {$\exanteprobability$};
% \draw (4.8, -0.8) node {$\correspondingquant$};
% \draw (9.8, 9) node {$\sfrac{\exanteprice}{\alpha}$};
% \draw (5, 9) node {$\exanteprice$};

\end{tikzpicture}
\end{minipage}
\begin{minipage}[t]{0.49\textwidth}
\begin{tikzpicture}[scale = 0.5]
%\fill[gray!40!white] (0, 0) -- (8, 0) -- (8, 7) -- (0, 7);

\draw (-0.2,0) -- (12.5, 0);
\draw (0, -0.2) -- (0, 6);

\draw plot [smooth, tension=0.6] coordinates { (0,0) (0.45, 1.4) (0.9,2.25)};
\draw plot [smooth, tension=0.6] coordinates { (0.9,2.25) (1.2, 2.15) (1.5, 1.9) (2,1.7) (3, 2.3)};
\draw plot [smooth, tension=0.6] coordinates {(3,2.3) (4,2.9) (5, 3)};
\draw plot [smooth, tension=0.6] coordinates {(5, 3) (7.5, 1.2) (10,1.1) (12,0)};
% \draw [dashed] plot [smooth, tension=0.6] coordinates { (0,0) (0.8,2) (2,2.6) (4, 3.28) (7.5, 1.2) (10,1.1) (12,0)};

\begin{scope}[very thick]

\draw plot [smooth, tension=0.6] coordinates { (0,0) (2,4) (8, 5) (12,0)};

\draw [gray, dashed] plot [smooth, tension=0.6] coordinates { (0,0) (0.45, 1.4) (0.9,2.25)};
\draw [gray, dashed] (0.9,2.25) -- (3,2.25);

% \draw [gray, dashed] plot [smooth, tension=0.6] coordinates { (0.9,2.25) (1.2, 2.15) (1.5, 1.9) (2,1.7) (3, 2.3)};
\draw [gray, dashed] plot [smooth, tension=0.6] coordinates {(3,2.3) (4,2.9) (5, 3)};
\draw [gray, dashed] plot (5, 3) -- (12, 3);
% \draw [gray, dashed] plot [smooth, tension=0.6] coordinates {(5, 3) (7.5, 1.2) (10,1.1) (12,0)};
\end{scope}

\draw [dotted] plot [smooth, tension=0.6] coordinates { (0,0) (2, 2) (8, 2.5) (12, 0) };
% \draw plot [smooth, tension=0.6] coordinates { (0,0) (2, 2.5) (8, 3.125) (12, 0) };

% \draw [gray, dotted] (0, 3.35) -- (12, 3.35);

\draw (0, -0.8) node {$0$};
\draw (12, -0.8) node {$1$};

\draw (3.6, 5.6) node {$\revcurve$};
\draw (8, 0.6) node {$\cumprice$};
\draw (8, 3.6) node {$\cumprice'$};
\draw (4, 1.8) node {$\sfrac{\revcurve}{\zeta}$};

\end{tikzpicture}
\end{minipage}
\end{flushleft}
\vspace{-10pt}
\caption{\label{f:definition} The left hand side illustrates an
  example that the price-posting revenue curve $\cumprice$ is
  $(\alpha, \beta)$-close for price posting to ex ante revenue curve
  $\revcurve$.  The definition requires that $\cumprice$ is at least
  $\revcurve / \alpha$ on $[0,1/\beta]$, i.e., that the black thin
  line is above the black dotted line.  The right hand side
  illustrates an example that the price-posting revenue curve
  $\cumprice$ is $\zeta$-close for ex ante optimization to ex ante
  revenue curve $\revcurve$.  Depicted with the gray dashed line,
  $\cumprice'$ is defined as $\cumprice'(\quant) = \max_{\quant' \leq
    \quant} \cumprice(\quant)$ and $\zeta$-closeness requires that
  $\cumprice' > \revcurve/\zeta$, i.e., that the gray dashed line is
  above the black dotted line.}
\end{figure}
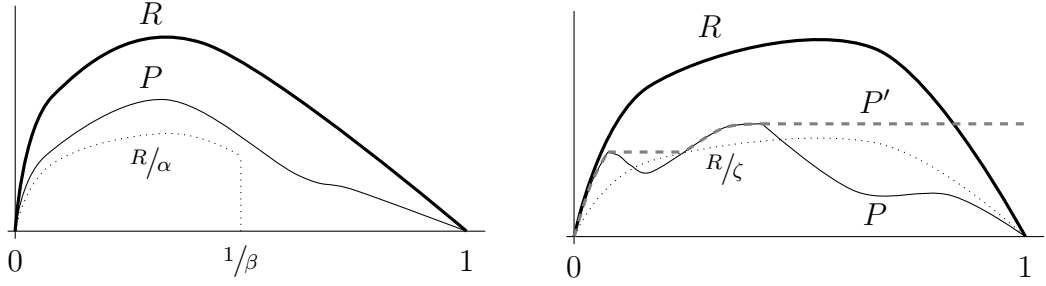

The bound in \Cref{thm:general AP bound} can be improved if we have
extra information about the closeness between the price posting
revenue curve and the ex ante revenue curve.  Specifically, let $\eta$
denote a bound on the ratio between the single agent optimal mechanism
and the single agent optimal price posting, i.e., $\max_\quant
\revcurve(\quant) / \max_\quant \cumprice(\quant) \leq \eta$.  It is
easy to see that if $\cumprice$ and $\revcurve$ are
$(\alpha,\beta)$-close then $\eta \leq \alpha\beta$.  Often, a better
bound on $\eta$ than $\alpha\beta$ can be obtained and, when it can,
\Cref{thm:improved AP bound} improves the bound of 
\Cref{thm:general AP bound}.
 
\begin{theorem}\label{thm:improved AP bound}
Given agents with 
ex ante revenue 
curves 
$\{\revcurve_i\}_{i=1}^n$
and price-posting revenue curves
$\{\cumprice_i\}_{i=1}^n$,
if each agent 
is $(\alpha, \beta)$-close for price posting
and 
the optimal price posting revenue 
for each agent
is an $\eta$-approximation to the  
optimal single agent revenue, 
then 
anonymous pricing on the price posting revenue curve  
is a $\improvedAP$-approximation
to anonymous pricing on the ex ante revenue curve,
i.e.,
$
\improvedAP
\cdot \pricerev(\{\cumprice_i\}_{i=1}^n)
\geq 
\pricerev(\{\revcurve_i\}_{i=1}^n)$. 
\end{theorem}

%\begin{proof}[Proof of \Cref{thm:improved AP bound}]
\begin{proof}
Let $\exanteprice$ be the anonymous price on ex ante revenue curve.
For a threshold $\threshold$ to be tuned subsequently, we say an agent
$i$ is a {\em significant contributor} to the anonymous pricing
revenue $\pricerev(\{\revcurve_i\}_{i=1}^n, \exanteprice)$ if
$\priceToQuantile(\exanteprice, \revcurve_i)$ is at least
$\sfrac{1}{\threshold}$.  We analyze separately the case that there is
a significant contributor and the case that all agents are
insignificant contributors.

\paragraph{When Significant Contributor Exists}
Suppose there exists an agent $i^*$ such that 
$\priceToQuantile(\exanteprice, \revcurve_{i^*})$
is at least $\sfrac{1}{\threshold}$.
Consider the optimal price-posting 
for this single agent $i^*$.
Its revenue is 
a $\threshold \eta \linearapproxratio$-approximation 
to the ex ante relaxation, i.e., 
\begin{align}
\label{eq:significant}
\pricerev(\{\cumprice_i\}_{i=1}^n)
\geq 
\max_{\quant} \cumprice_{i^*}(\quant) 
\geq \frac{1}{\eta} \max_{\quant}
\revcurve_{i^*}(\quant) 
\geq \frac{1}{\eta \threshold}
\pricerev(
\{\revcurve_i\}_{i=1}^n, \exanteprice) 
\end{align} 
where the first inequality is 
due to our $\eta$-approximation assumption
for a single agent,
and the second inequality 
is 
due to the fact 
that agent $i^*$ is a significant contributor.
\paragraph{When Significant Contributor Does Not Exist}
Suppose there is no significant contributor,
i.e.,
$\priceToQuantile(\exanteprice, \revcurve_{i}) 
< \sfrac{1}{\threshold}$ for all $i$.
By a similar argument to
\Cref{thm:general AP bound}, 
we can lower bound 
$\priceToQuantile(\frac{\exanteprice}{\alpha}, \cumprice_i)$
for every agent $i$ as follows,
\begin{align*}
\notag
    \priceToQuantile
    \left(\frac{\exanteprice}{\alpha}, \cumprice_i
    \right)
    \geq 
    \left\{
    \begin{array}{ll}
      \priceToQuantile(\exanteprice, \revcurve_i)
      &
      \text{ if }
      \priceToQuantile(\exanteprice, \revcurve_{i}) 
\leq \sfrac{1}{\beta}\\
      \frac{\threshold}{\beta} \cdot 
\priceToQuantile(\exanteprice, \revcurve_i)
     & \text{ otherwise}
    \end{array}
    \right.
\end{align*}
where  the latter case uses the facts that 
$\priceToQuantile
\left(\frac{\exanteprice}{\alpha}, \cumprice_i \right)
\geq \sfrac{1}{\beta}$ 
and 
$\priceToQuantile(\exanteprice, \revcurve_{i}) 
< \sfrac{1}{\threshold}$.

Therefore, for parameter $\threshold \leq \beta$, 
the the revenue of posting price $\frac{\exanteprice}{\alpha}$ 
is 
\begin{align}
\notag
    \pricerev\left(\{\cumprice_i\}_{i=1}^n,
    \frac{\exanteprice}{\alpha}\right) 
    &= \frac{\exanteprice}{\alpha} \cdot 
    \left(1 - \prod\nolimits_i (1-\priceToQuantile(\frac{\exanteprice}{\alpha}, \cumprice_i)) \right) 
    \geq \frac{\exanteprice}{\alpha} \cdot 
    \left(1 - \prod\nolimits_i (1 - \frac{\threshold}{\beta}
    \priceToQuantile(\exanteprice, \revcurve_i)) \right) \\
\label{eq:all-insignificant}
    &\geq \frac{\exanteprice \threshold}{\alpha\beta} 
    \cdot \left(1 - \prod\nolimits_i 
    (1 - \priceToQuantile(\exanteprice, \revcurve_i)) \right) 
    = \frac{\threshold}{\alpha\beta}
    \pricerev(\{\revcurve_i\}_{i=1}^n, \exanteprice). 
\end{align}

Notice that the approximation bounds of
equations~\eqref{eq:significant} and~\eqref{eq:all-insignificant} have
$\threshold$ in the denominator and numerator, respectively.  Thus,
setting $\threshold$ to equalize these bounds (when possible) will
optimize approximation bound.  Specifically, if $\alpha \leq \beta
\eta$, setting $\threshold = \sqrt{\sfrac{\alpha \beta}{\eta}}$
optimizes the approximation bound as $\sqrt{\alpha \beta \eta}$; and
if $\alpha > \beta \eta$, setting $\threshold = \beta$ optimizes the
approximation bound as $\alpha$.
\end{proof}

\Cref{thm:general AP bound} and \Cref{thm:improved AP bound} reduce
the approximation of anonymous pricing for non-linear agents to that
of regular linear agents.  For example, if $\linearapproxratio \cdot
\pricerev(\{\revcurve_i\}_{i=1}^n) \geq
\exanterev(\{\revcurve_i\}_{i=1}^n)$, \Cref{thm:general AP bound}
implies $\alpha\beta\linearapproxratio \cdot
\pricerev(\{\cumprice_i\}_{i=1}^n) \geq
\exanterev(\{\revcurve_i\}_{i=1}^n)$.
%, and \Cref{thm:improved AP bound}
%implies that $\improvedAP \cdot \linearapproxratio \cdot
%\pricerev(\{\cumprice_i\}_{i=1}^n) \geq
%\exanterev(\{\revcurve_i\}_{i=1}^n)$.  
By \Cref{lem:ex-ante is concave}, the ex ante revenue curves
$\{\revcurve_i\}_{i=1}^n$ are concave and thus, by 
\Cref{thm:linear bound}, the approximation ratio $\gamma$ is at most
$\rho \approx e$.

\begin{corollary}
\label{cor:reduction framework 1}
Given agents with 
ex ante revenue 
curves 
$\{\revcurve_i\}_{i=1}^n$
and price-posting revenue curves
$\{\cumprice_i\}_{i=1}^n$, if each agent is 
$(\alpha, \beta)$-close for price posting,
then the worst case approximation factor of 
anonymous pricing to the ex ante relaxation 
is
$\alpha \beta \rho$, 
i.e., 
$\alpha\beta\rho \cdot 
\pricerev(\{\cumprice_i\}_{i=1}^n) 
\geq \exanterev(\{\revcurve_i\}_{i=1}^n)$. 
If additionally, 
the optimal price posting revenue 
for each agent 
is an $\eta$-approximation to the  
optimal single agent revenue, 
then the worst case approximation factor of 
anonymous pricing to the ex ante relaxation 
is
$\improvedAP \cdot \rho$, 
i.e., 
$\improvedAP \cdot \rho \cdot 
\pricerev(\{\cumprice_i\}_{i=1}^n) 
\geq \exanterev(\{\revcurve_i\}_{i=1}^n)$. 
\end{corollary}

The reduction framework
in \Cref{cor:reduction framework 1}
is tight with the given assumptions. 
See Appendix \ref{apx:tight framework} 
for a lower bound instance.

% In Appendix \ref{apx:tight framework}, 
% we show that this reduction framework is tight with the given assumptions. 

\subsection{Approximate Robustness of the Ex Ante Relaxation}

We develop the lower path of \Cref{f:framework} for reducing the
non-linear anonymous pricing problem to the linear anonymous pricing
problem.  This lower path is based on identifying a closeness property
on revenue curves that approximately preserves the optimal revenue of
the ex ante relaxation.  The revenue of ex ante relaxation is
approximately preserved under weaker conditions than the revenue of
anonymous pricing.  On the other hand, to make use of this approximate
robustness to derive an anonymous pricing approximation, 
our framework
additionally requires the assumption that non-linear agents'
price-posting revenue curves are concave.

\begin{definition}
The agent's price-posting revenue curve $\cumprice$ is
\emph{$\zeta$-close for ex ante optimization} to her ex ante revenue
curve $\revcurve$, if for all $\quant \in [0, 1]$, there exists a
quantile $\correspondingquant \leq \quant$ such that
$\cumprice(\correspondingquant) \geq \frac{1}{\zeta}
\revcurve(\quant)$.  Such an agent is \emph{$\zeta$-close for ex ante
  optimization}.
\end{definition}

The above definition is illustrated in
the right hand side of \Cref{f:definition}. 
Note that $(\alpha, \beta)$-close for price posting 
implies $(\alpha\beta)$-close for ex
ante optimization. 

\begin{theorem}
\label{thm:general EAR bound}
Given agents with ex ante revenue curves $\{\revcurve_i\}_{i=1}^n$ and
price-posting revenue curves $\{\cumprice_i\}_{i=1}^n$, if each agent
$i$ is $\zeta$-close for ex ante optimization, then the optimal ex
ante relaxation on revenue curves $\{\cumprice_i\}_{i=1}^n$ is a
$\zeta$-approximation to the optimal ex ante relaxation on revenue
curves $\{\revcurve_i\}_{i=1}^n$, i.e.,
$\exanterev(\{\cumprice_i\}_{i=1}^n) \geq \frac{1}{\zeta}
\exanterev(\{\revcurve_i\}_{i=1}^n)$.
\end{theorem}
\begin{proof}
Let $\{\quant_i\}_{i=1}^n$ be the optimal ex ante relaxation for ex
ante revenue curves $\{\revcurve_i\}_{i=1}^n$, and
let~$\correspondingquant_i$ be the quantile assumed to exist by
$\zeta$-closeness such that $\correspondingquant_i \leq \quant_i$ and
$\cumprice_i(\correspondingquant_i) \geq
\frac{1}{\zeta}\revcurve_i(\quant_i)$ for each $i$.  We have
\begin{equation*}
    \exanterev(\{\cumprice_i\}_{i=1}^n) 
    \geq \sum\nolimits_i \cumprice_i(\correspondingquant_i) 
    \geq \frac{1}{\zeta}\, \sum\nolimits_i \revcurve(\quant_i) 
    = \frac{1}{\zeta}\, \exanterev(\{\revcurve_i\}_{i=1}^n). \qedhere
\end{equation*}
\end{proof}

When the price posting revenue curves are concave, 
\Cref{thm:general EAR bound} can be combined with 
\Cref{thm:linear bound} applied to
$\{\cumprice_i\}_{i=1}^n$ to obtain the following corollary.

\begin{corollary}\label{cor:ex ante close}
%For the single item environment, 
Given agents with ex ante revenue curves $\{\revcurve_i\}_{i=1}^n$ and
price-posting revenue curves $\{\cumprice_i\}_{i=1}^n$, if each agent
has concave price-posting revenue curve and is $\zeta$-close for ex
ante optimization, then the worst case approximation factor of
anonymous pricing to the ex ante relaxation is $\zeta \rho$, i.e.,
$\zeta \rho \cdot \pricerev(\{\cumprice_i\}_{i=1}^n) \geq
\exanterev(\{\revcurve_i\}_{i=1}^n)$.
\end{corollary}

\subsection{Heterogeneous Agent Utility Models}

Our closeness definitions are monotonic, formalized in the subsequent
lemma.  With this observation, our framework can be applied to
environments with heterogeneous utility functions.  For example,
suppose some of the agents have private budget constraints and some of
the agents are risk averse.  If each agent $i \in \{1,\ldots,n\}$ is
$(\alpha_i, \beta_i)$-close for price posting, then anonymous pricing
for these agents is a $(\max\{\alpha_i\} \cdot \max\{\beta_i\} \cdot
\rho)$-approximation to the optimal ex ante relaxation.

\begin{lemma}\label{lem:closeness is downward imply}
For any $\alpha' \geq \alpha \geq 1$, $\beta' \geq \beta \geq 1$, and
$\zeta' \geq \zeta \geq 1$, 
% $(\alpha, \beta)$-regularity implies
% $(\alpha', \beta')$-regularity, 
$(\alpha, \beta)$-close for price posting
implies $(\alpha', \beta')$-close for price posting, 
and $\zeta$-closeness for ex ante optimization 
implies $\zeta'$-close for ex ante optimization.
\end{lemma}

% \begin{itemize}
%     \item $(\alpha, \beta)$-regular implies $(\alpha', \beta')$-regular; 

%     \item 
%     $(\alpha, \beta)$-close for price posting 
%     implies $(\alpha', \beta')$-close for price posting;
    
%     \item 
%     $\zeta$-close for ex ante relaxation 
%     implies $\zeta'$-close for ex ante relaxation.
% \end{itemize}

\section{Public-budget Utility}

\label{sec:public}
In this section, we consider 
the case where 
agents have public-budget utility.
% We make the following two standard
% assumptions on the valuation distributions.
% 
% \begin{assumption}\label{asp:value regular public}
% For each agent $i$,
% her valuation distribution 
% $F_i$ is regular.
% \end{assumption}
% 
% 
% Though in our model, the agents have non-identical 
% valuation distributions,
% since our reduction framework 
% only requires the 
% comparison between the ex ante revenue curve 
% and the price-posting revenue curve for each agent individually,
% we drop the index $i$ for agent $i$ for notation simplicity.
% 
% 
% \begin{assumption}\label{asp:value density decreasing}
% For each agent,
% her valuation $F$ has 
% decreasing density
% on its support $[0,\hval]$.
% \end{assumption}
% 
In the linear utility model,
\citet{AHNPY-18} show that 
the regularity on valuation distribution
is necessary for the  constant approximation
bound
of the anonymous pricing,
so we make the same assumption under 
the public-budget utility.
%\Cref{asp:value regular public}, \ref{asp:value density decreasing} together are called 
% We also consider an additional assumption that the density function is
% decreasing on its support.  These two assumptions together are
% sometimes called \emph{public-budget regularity} %assumption
% \citep[cf.][]{PV-14}. 
The following theorem and corollary summarize
the main results of this section.

% \begin{comment}
% \begin{theorem}
% \label{thm:public e bound}
% For a single item 
% environment with 
% agents with public-budget utility,
% under \Cref{asp:value regular public},
% \ref{asp:value density decreasing},
% the worst case
% approximation factor of
% anonymous pricing to the ex ante
% relaxation is at most $\rho$.
% \end{theorem}
% \end{comment}

% \begin{theorem}
% \label{thm:public 2e bound}
% An agent with public-budget utility and regular valuation distribution
% is $(2, 1)$-close for price posting, and with regular and decreasing
% density valuation distribution the agent is $(1, 1)$-close for price posting.
% \end{theorem}

\begin{theorem}
\label{thm:public e bound imporved}
An agent with public-budget utility and regular valuation distribution
is $(1, 1)$-close for price posting.
\end{theorem}

\begin{corollary}
\label{cor:public e bound}
For a single item 
environment with 
agents with public-budget utility
and regular valuation distributions,
%under \Cref{asp:value regular public},
the worst case
approximation factor of
anonymous pricing to the ex ante
relaxation is at most $\rho$;
with regular valuation
distributions. 
%additional \Cref{asp:value density decreasing},
% the worst case
% approximation factor of
% anonymous pricing to the ex ante
% relaxation is at most $\rho$.
\end{corollary}

% To prove 
% %\Cref{thm:public e bound}
% %and 
% \Cref{thm:public e bound imporved},
% we start with a characterization of 
% the ex ante optimal mechanism.

% \begin{lemma}
% \label{lem:public budget single price optimal imporved}
% %% For a single agent with public-budget utility and regular valuation
% %% distribution, the $\quant \in [0,1]$ ex ante optimal mechanism is a
% %% menu with at most two options; with regular and decreasing density
% %% valuation distribution the $\quant$ ex ante optimal mechanism is a
% %% per-unit pricing.
% For a single agent with public-budget utility and regular valuation distribution, the $\quant \in [0,1]$ ex
% ante optimal mechanism is a per-unit pricing, i.e., $\revcurve(\quant) = \cumprice(\quant)$.  
% \end{lemma}
 
To prove \Cref{thm:public e bound imporved},
it is sufficient to show for any quantile 
$\exquant\in[0, 1]$, 
the $\exquant$ ex ante optimal mechanism
is a price-posting mechanism,
i.e., $\revcurve(\exquant) = \cumprice(\exquant)$.
To show this, 
we write the ex ante optimal mechanism 
as an optimization program,
and 
apply Lagrangian relaxation on
the budget constraint.
This leads to a new optimization
program similar to an agent with linear 
utility but with a 
Lagrangian objective function.
Following the technique 
that price-posting revenue curve
indicates the ex ante
optimal mechanism
for a linear agent,
we 
consider 
the \emph{Lagrangian price-posting 
revenue curve} 
which characterizes the ex ante
optimal mechanism for the 
Lagrangian objective function.
See further discussion about this 
technique in \citet{AFHH-13}
and \citet{FH-18}.
The detailed proof of 
\Cref{thm:public e bound imporved}
% uses an argument based on the
% Lagrangian relaxation of the budget constraint,
% which we 
is deferred to  
\Cref{apx:public budget regular}.
By comparing \Cref{thm:public e bound imporved}
with \Cref{thm:linear bound}, 
we know that the worst ratio for 
comparing anonymous pricing and ex ante relaxation
happens when the budgets never bind.

\section{Independent Private-budget Utility}
\label{sec:private}
In this section, we consider the case where agents have private-budget
utility.  To obtain a constant approximation bound for anonymous
pricing, we assume that values and budgets are independently
distributed, the valuation distribution is regular and one of two
possible assumptions on the budget distribution: (a) the probability
that the budget exceeds its expectation is at least a constant
$\sfrac{1}{\budgetQuantile}$; and (b) the budget distribution has
\emph{monotone hazard rate (MHR)}, i.e,
$\frac{g(\wealth)}{1-G(\wealth)}$ is monotonically non-decreasing with
respect to $\wealth$.  Since for a MHR distribution, a sample exceeds
its expectation with probability at least $\sfrac{1}{e}$
\citep[cf.][]{BM-65}, assumption (b) is special case of assumption (a)
with $\budgetQuantile = e$.  We give two examples in \Cref{apx:private
  negative example} that demonstrate the necessity of these
assumptions to guarantee the constant approximation of anonymous
pricing.

\subsection{Unconstrained 2-approximation}

Before giving results about the closeness of ex ante revenue curves
and price-posting revenue curves, we first show that the optimal
pricing is a 2-approximation to the optimal mechanism (with no ex ante
constraint).  With this result we will be able to invoke
\Cref{thm:improved AP bound} with $\eta = 2$.  Though it will not
result in any improvements in our main results, this approximation
bound does not make any assumptions on the budget distribution.

\begin{theorem}
\label{thm:single agent approx private}
For selling a single item to a single private-budget agent with
independent value and budget distributions and regular value
distribution, the optimal per-unit pricing is a 2-approximation to the optimal
mechanism.
\end{theorem}

We give the proof of this theorem below.  It follows from a
characterization of the optimal revenue for public-budget agents from
\citet{LR-96} and a generalization of a geometric argument from
\citet{DRY-10} that analyzes the revenue from a random price from the
value distribution, both stated below.  We employ this random-price
argument later in \Cref{lem:posting randomized price for EXL} as well.

\begin{lemma}[\citealp{LR-96}]
\label{lem:public budget optimal under regular}
For a single item and single public-budget agent with regular
valuation distribution, the optimal mechanism posts a price equal to
the smaller of the agent's budget and the monopoly reserve.
\end{lemma}

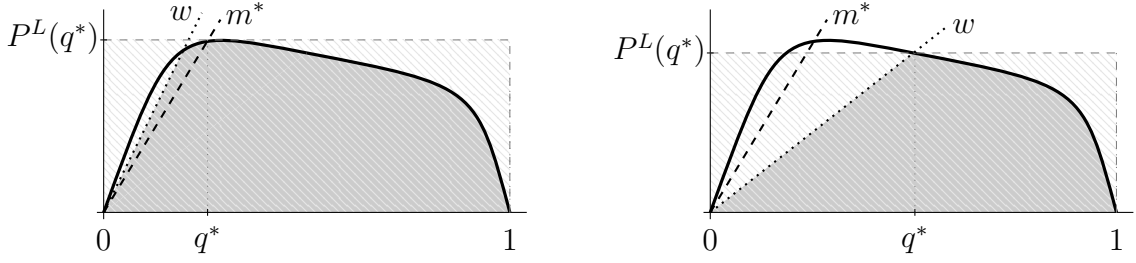
\begin{figure}[t]
\begin{flushleft}
\hspace{-5pt}
\begin{minipage}[t]{0.48\textwidth}
\centering
\begin{tikzpicture}[scale = 0.45]

% \draw [name path = E, white] (6, 
% 4.714285714285714) -- (0, 0);

\draw [name path = E, white] (0, 0) .. controls 
(2, 5.5) .. (6, 4.714285714285714);

\draw [name path global = B, white]
(6, 4.714285714285714) .. controls 
(11, 3.75) .. (12, 0);

\fill[color=gray!40!white] (0, 0) -- (6, 4.714285714285714) 
-- (12, 0);

\tikzfillbetween[of=B and E]{gray!40!white};

\draw[color=gray!40!white]
(6, 4.714285714285714) -- (12,0);
\fill[color=white] (0, 0) -- (2.9, 5.9)
-- (0, 5.9) -- (0, 0);

\fill[pattern=north west lines, pattern color=gray!20!white] (0, 0) -- (0, 5.1) 
-- (12, 5.1) -- (12, 0);

\draw (-0.2,0) -- (12.5, 0);
\draw (0, -0.2) -- (0, 6);

\begin{scope}[very thick]

% \draw [dotted] (6.7, 0) -- (6.7, 7.3);
% \draw [dotted] (4.8, 0) -- (4.8, 4.8);
% \draw [dotted] (3.4, 0) -- (3.4, 6.8);

%\draw plot [smooth, tension=0.6] coordinates { (0,0) (1,3.6) (4, 5) (8,4) (12,0)};

\draw [name path global = A](0, 0) .. controls 
(2, 5.5) .. (6, 4.714285714285714);

\draw [name path global = F]
(6, 4.714285714285714) .. controls 
(11, 3.75) .. (12, 0);

\end{scope}

%\draw (0, 5.1) -- (12, 5.1);

\draw [name path = C, dashed, thick] (0, 0) -- (3.45, 5.7);
\draw [name path = D, dotted, thick] (0, 0) -- 
(2.9, 5.9);

\draw (0, -0.8) node {$0$};
\draw (12, -0.8) node {$1$};

\draw (4.2, 5.9) node {$\myersonReserve$};
\draw (2.3, 5.9) node {$\wealth$};

\draw [gray, dashed] (0, 5.1) -- (12, 5.1);
\draw [gray, dashed] (12, 0) -- (12, 5.1);
\draw (-1.5, 5.2) node
{$\cumpricelinear(\quant^*)$};
\draw (-0.1, 5.1) -- (0.1, 5.1);

\draw [gray, dotted] (3.07, 5.1) -- (3.07, 0);
\draw (3.07, 0.1) -- (3.07, -0.1);

\draw (3.07, -0.8) node{$\quant^*$};

%\draw [gray, dotted] (9, 4.05) -- (12, 4.05);
\draw [gray, dotted] (12, 0) -- (12, 4.05);

\end{tikzpicture}

% \begin{tikzpicture}[scale = 0.5]

% \draw (-0.2,0) -- (12.5, 0);
% \draw (0, -0.2) -- (0, 6);

% \begin{scope}[very thick]

% % \draw [dotted] (6.7, 0) -- (6.7, 7.3);
% % \draw [dotted] (4.8, 0) -- (4.8, 4.8);
% % \draw [dotted] (3.4, 0) -- (3.4, 6.8);

% \draw plot [smooth, tension=0.6] coordinates { (0,0) (1,3.6) (4, 5) (8,4) (12,0)};

% % \draw [dashed] plot [smooth, tension=0.6] coordinates { (0,0) (0.8,2) (4, 3.5) (7.5, 1.5) (9,1.1) (12,0)};
% \draw (0, 0) -- (12, 3);
% \draw (0, 0) -- (7, 5.5);
% % \draw (0, 0) -- (4.5, 9);

% \end{scope}

% \draw (0, -0.8) node {$0$};
% \draw (12, -0.8) node {$1$};

% \draw (12, 3.5) node {$\price$};
% \draw (7.5, 5.5) node {$\wealth$};

% \draw [gray, dotted] (0, 4.8) -- (6, 4.8);
% \draw (-2.5, 4.8) node {$\Rev[\wealth]{\OPT}$};

% \draw [gray, dotted] (9.7, 2.4) -- (12, 2.4);
% \draw [gray, dotted] (12, 0) -- (12, 2.4);

% \end{tikzpicture}
\end{minipage}
\begin{minipage}[t]{0.48\textwidth}
\centering
\begin{tikzpicture}[scale = 0.45]

\draw [name path = E, white] (6, 4.714285714285714) -- (0, 0);

\draw [name path global = B, white]
(6, 4.714285714285714) .. controls 
(11, 3.75) .. (12, 0);

\fill[color=gray!40!white] (0, 0) -- (6, 4.714285714285714) 
-- (12, 0);
\draw[color=gray!40!white]
(6, 4.714285714285714) -- (12,0);
% \fill[color=gray!40!white] (9, 4.05) -- (12, 4.05)
% -- (12, 0) -- (9, 0);

\tikzfillbetween[of=B and E]{gray!40!white};

\fill[pattern=north west lines, pattern color=gray!20!white] (0, 0) -- (0, 4.714285714285714) 
-- (12, 4.714285714285714) -- (12, 0);

\draw (-0.2,0) -- (12.5, 0);
\draw (0, -0.2) -- (0, 6);

\begin{scope}[very thick]

% \draw [dotted] (6.7, 0) -- (6.7, 7.3);
% \draw [dotted] (4.8, 0) -- (4.8, 4.8);
% \draw [dotted] (3.4, 0) -- (3.4, 6.8);

%\draw plot [smooth, tension=0.6] coordinates { (0,0) (1,3.6) (4, 5) (8,4) (12,0)};

\draw [name path global = A](0, 0) .. controls 
(2, 5.5) .. (6, 4.714285714285714);

\draw [name path global = F]
(6, 4.714285714285714) .. controls 
(11, 3.75) .. (12, 0);

\end{scope}

%\draw (0, 5.1) -- (12, 5.1);

\draw [name path = C, dashed, thick] (0, 0) -- (3.45, 5.7);
\draw [name path = D, dotted, thick] (0, 0) -- (7, 5.5);

\draw (0, -0.8) node {$0$};
\draw (12, -0.8) node {$1$};

\draw (4.2, 5.9) node {$\myersonReserve$};
\draw (7.5, 5.5) node {$\wealth$};

\draw [gray, dashed] (0, 4.714285714285714) -- (12, 4.714285714285714);
\draw [gray, dashed] (12, 0) -- (12, 4.714285714285714);
\draw (-1.5, 4.8) node
{$\cumpricelinear(\quant^*)$};
\draw (-0.1, 4.714285714285714) -- (0.1, 4.714285714285714);

\draw [gray, dotted] 
(6.045, 4.714285714285714) -- (6.045, 0);
\draw (6.045, 0.1) -- (6.045, -0.1);

\draw (6.045, -0.8) node{$\quant^*$};

%\draw [gray, dotted] (9, 4.05) -- (12, 4.05);
\draw [gray, dotted] (12, 0) -- (12, 4.05);

\end{tikzpicture}

% \begin{tikzpicture}[scale = 0.5]

% \draw (-0.2,0) -- (12.5, 0);
% \draw (0, -0.2) -- (0, 6);

% \begin{scope}[very thick]

% % \draw [dotted] (6.7, 0) -- (6.7, 7.3);
% % \draw [dotted] (4.8, 0) -- (4.8, 4.8);
% % \draw [dotted] (3.4, 0) -- (3.4, 6.8);

% \draw plot [smooth, tension=0.6] coordinates { (0,0) (1,3.6) (4, 5) (8,4) (12,0)};

% % \draw [dashed] plot [smooth, tension=0.6] coordinates { (0,0) (0.8,2) (4, 3.5) (7.5, 1.5) (9,1.1) (12,0)};
% \draw (0, 0) -- (12, 3);
% \draw (0, 0) -- (7, 5.5);
% % \draw (0, 0) -- (4.5, 9);

% \end{scope}

% \draw (0, -0.8) node {$0$};
% \draw (12, -0.8) node {$1$};

% \draw (12, 3.5) node {$\price$};
% \draw (7.5, 5.5) node {$\wealth$};

% \draw [gray, dotted] (0, 4.8) -- (6, 4.8);
% \draw (-2.5, 4.8) node {$\Rev[\wealth]{\OPT}$};

% \draw [gray, dotted] (9.7, 2.4) -- (12, 2.4);
% \draw [gray, dotted] (12, 0) -- (12, 2.4);

% \end{tikzpicture}
\end{minipage}
\end{flushleft}
\caption{\label{f:revenue of sample 0} In the geometric proof of
  \Cref{lem:single agent approx public}, for any fixed budget
  $\wealth$ the optimal revenue is the area of the light gray striped
  rectangle and the revenue from posting random price $\randomprice$
  is the area of the dark gray region.  By geometry, the latter is at
  least half of the former.  The black curve is the price-posting
  revenue curve with no budget constraint $\cumpricelinear$.  The
  figure on the left depicts the large-budget case (i.e., $\wealth \geq
  \myersonReserve$), and the figure on the right depicts the
  small-budget case (i.e., $\wealth \leq \myersonReserve$).}
\end{figure}

The following lemma generalizes a lemma and proof approach by
\citet{DRY-10} to public-budget agents.

\begin{lemma}
\label{lem:single agent approx public}
For a single item and a single public-budget agent with regular
valuation distribution, posting a random per-unit price drawn from the
valuation distribution is a $2$-approximation to the optimal revenue.
\end{lemma}

\begin{proof} 
Denote the agent's fixed budget by $\wealth$ and let $\randomprice$ be
a random per-unit price drawn from valuation distribution $F$.  For
the depiction of \Cref{f:revenue of sample 0}, we claim that the dark
grey shaded area is the expected revenue of posting the random price
$\randomprice$, and the area of the light grey striped rectangle is
the optimal revenue.  Concavity of the price-posting revenue curve
with no budget constraint (by regularity of the value distribution)
implies that a triangle with half the area of the light gray rectangle
is contained within the dark gray region and, thus, the random price
is a 2-approximation.  The remainder of this proof shows that the
geometry of the regions described above is correct.

Let $\cumpricelinear$ be the price-posting revenue curve with no
budget constraint.  When $\randomprice \leq \wealth$, the revenue of
posting price $\randomprice$ is $\cumpricelinear(1-F(\randomprice))$
and, when $\randomprice > \wealth$, the revenue of posting price
$\randomprice$ is $\wealth\,(1-F(\randomprice))$.  Let $\quant^* =
\max\{1-F(\wealth), 1-F(\myersonReserve)\}$.  \Cref{lem:public budget
  optimal under regular} implies that the optimal revenue for budget
$\wealth$ is $\cumpricelinear(\quant^*)$ by posting the minimum
between the monopoly reserve $\myersonReserve$ and the budget
$\wealth$, i.e., it is the area of the light gray striped rectangle.
On the other hand, the revenue of posting $\randomprice$ is
$\int_0^{1-F(\wealth)} \wealth\,\quant\,d\quant + \int_{1-F(\wealth)}^1
\cumpricelinear(\quant)\,d\quant $, which is exactly the dark grey
shaded area.
\end{proof}

\begin{proof}[Proof of \Cref{thm:single agent approx private}]
Consider the revenue from types with each fixed budget level
separately.  At each budget level, \Cref{lem:single agent approx
  public} guarantees that the revenue from posting a random per-price
is a 2-approximation to the optimal revenue from these types.  Since
value and budget are independent, the random per-unit prices offered
at all budget levels are identically distributed, thus simultaneously
offering the same random per-unit price across all budget levels
guarantees a 2-approximation to optimal revenue for the private-budget
agent.  Furthermore, the optimal deterministic per-unit price obtains
revenue that is no worse than that of the randomized price.
\end{proof}

\subsection{Budgets Exceeding Expectation with 
Constant Probability}

We first analyze the approximation bound 
for anonymous pricing mechanism
under the assumption that 
the budget exceeds its 
expectation with constant probability 
at least $\sfrac{1}{\budgetQuantile}$.

\begin{theorem}\label{thm:private bound}
A single private-budget agent is $(\alphaprivate,\betaprivate)$-close
for price posting if value and budget are independently distributed,
the valuation distribution is regular, and the budget exceeds its
expectation with probability at least $\sfrac{1}{\budgetQuantile}$.
\end{theorem}

With \Cref{thm:single agent approx private}
and \Cref{thm:private bound},
\Cref{cor:reduction framework 1}
implies that under the same assumptions
in \Cref{thm:private bound},
the approximation factor of 
anonymous pricing is 
$\APprivate$.

\begin{corollary}
For a single item environment with private-budget agents, anonymous
pricing is a $\APprivate$-approximation to the ex ante relaxation if
values and budgets are independently distributed, the valuation
distributions are regular, and the budgets exceed their expectations
with probability at least $\sfrac{1}{\budgetQuantile}$.
\end{corollary}

% To prove \Cref{thm:private bound},
% it is sufficient to show that 
% the private-budget utility under our assumptions 
% has price-posting revenue curve 
% which is $(\alphaprivate,\betaprivate)$-close for price posting, 
% and then apply \Cref{lem:close rev curve imply close pricing}.

We take the following high-level approach which is similar to the
analysis of \citet{abr-06}. Fix any ex ante constraint $\quant$ less
than $\sfrac{1}{(\betaprivate)}$.  Consider the per-unit price induces
ex ante allocation probability exactly $\quant$; henceforth, the {\em
  market clearing price}.  For any allocation and payment
$(\alloc,\price)$ obtained by an agent in the ex ante optimal
mechanism, refer to the agent's per-unit price as $\tilde\price =
\sfrac{\price}{\alloc}$.  We decompose the $\quant$ ex ante optimal
mechanism into two mechanisms where the per-unit price in the first
mechanism for each type is at most the market clearing price, and the
per-unit price in the second mechanism for each type is at least the
market clearing price.  The revenue of the former mechanism is smaller
than the revenue of posting the market clearing price since it sells
less of the item and with a lower per-unit price and, under our
assumptions, the revenue of the latter can be bounded within a
constant factor of the revenue of posting the market clearing price.
The latter analysis considers separately the case of high and low
market clearing prices, and in the former derives bounds via the
geometry of revenue curves and in the latter by the assumption on
budgets.
% using the observation that, 
% under our assumption on the budget distribution, 
% either posting the market clearing price is a constant approximation 
% to the optimal revenue even without the budget constraint, 
% or the budget binds in posting the market clearing price 
% for types with budget equal to the expected budget,
% and this induces a lower bound of the revenue 
% in posting the market clearing price.

% Finally, depending on 
% whether the market clearing price is large or small,
% we use two different arguments
% to 
% upper bound the revenue from 
% both decomposed mechanisms 
% by the revenue from 
% posting the market clearing price.
% %We now formalize this plan.

In order to decompose 
the ex ante optimal mechanism,
we first provide a definition
and a characterization 
of all incentive compatible
mechanisms for 
a single agent with private-budget utility,
and her behavior in the mechanisms.

\begin{definition}
An \emph{allocation-payment function} 
$\APF:[0,1] \rightarrow \R_+$
is a mapping from the allocation $\alloc$
purchased by an agent
to the payment $\price$ she is charged.
\end{definition}

\begin{lemma}\label{lem:pricing function characterization}
For a single private-budget agent, any incentive compatible mechanism,
and all types with any fixed budget; the mechanism provides a convex
and non-decreasing allocation-payment function, and subject to this
allocation-payment function, each type will purchase as much as she
wants until the budget constraint binds, the unit-demand constraint
binds, or the value binds (i.e., her marginal utility becomes zero).
\end{lemma}
The formal proof of \Cref{lem:pricing function characterization} is
given in \Cref{apx:private}.  At a high level, in private
budget settings, by relaxing the incentive compatibility constraint for
both values and budgets to the incentive compatibility constraint only for
values, the payment identity at every fixed budget level is sufficient
to induce an allocation-payment function for all types with that fixed
budget.
% A simple but crucial lemma in our analysis 
% is that subject to a fixed allocation-payment function,
% the revenue from types with higher budget can be upperbounded
% by the revenue from types with lower budget.

% the payment identity induces an allocation-payment function.
% In private-budget utility,  
% the incentive compatible constraints for value 
%  are
% already sufficient to pin down the same payment identity,
% and induce an allocation-payment function 
% for all types with that fixed budget.

Fix an arbitrary ex ante constraint $\quant$ and denoting $\EX$ be the
$\quant$ ex ante optimal mechanism.  Consider the decomposition of
$\EX$ into two mechanisms $\EXS$ and $\EXL$. These decomposed
mechanisms will be incentive compatible for values, but not
necessarily for budgets; a setting we refer to as the
\emph{random-public-budget utility} model.  We require that (a) the
per-unit prices in $\EXS$ for all types are at most the market
clearing price; (b) the allocation for all types in $\EXS$ is at most
the allocation in $\EX$; (c) the per-unit prices in $\EXL$ for all
types are larger than the market clearing price; and (d) the revenue
from $\EX$ (for a private-budgeted agent) is upper bounded by the
revenue from $\EXS$ and $\EXL$ (for a random-public-budgeted agent
with same distribution), i.e., $\Rev{\EX} \leq \Rev{\EXS} +
\Rev{\EXL}$.

The construction is as follows:
For each budget $\wealth$, 
let $\APF_\wealth$ be the allocation-payment function
for types with budget $\wealth$
in mechanism $\EX$, and 
$\alloc^*_\wealth$ be the largest allocation such that
the marginal price is below the market clearing price $\marketclearing$, i.e.,
$\alloc^*_\wealth =\argmax\{\alloc : \APF_\wealth'(\alloc )\leq \marketclearing\} $.
Define the allocation-payment functions 
$\APFS_\wealth$ and 
$\APFL_\wealth$ for 
$\EXS$ and $\EXL$ respectively below,
\begin{align*}
    \APFS_\wealth(\alloc) &=
    \left\{
    \begin{array}{ll}
      \APF_\wealth(\alloc)   &   
    \text{if }\alloc \leq \alloc^*_\wealth,\\
        \infty   &  \text{otherwise};
    \end{array}
    \right.
    \quad
    \APFL_\wealth(\alloc) = 
    \left\{
    \begin{array}{ll}
      \APF_\wealth(\alloc^*_\wealth+\alloc) -
      \APF_\wealth(\alloc^*_\wealth)&   
    \text{if }\alloc \leq 1 - \alloc^*_\wealth,\\
        \infty   &  \text{otherwise}.
    \end{array}
    \right.
\end{align*}

\begin{figure}[t]
\begin{flushleft}
\hspace{-5pt}
\begin{minipage}[t]{0.42\textwidth}
\centering
\begin{tikzpicture}[scale = 0.5]

\draw (-0.2,0) -- (11.5, 0);
\draw (0, -0.2) -- (0, 7.2);

\begin{scope}[ultra  thick]

% \draw [dashed] plot [smooth, tension=0.7] coordinates { (0,0) (3, 0.8) (6, 2)};
% \draw [dashed] plot [smooth, tension=0.8] coordinates {(6, 2) (9, 4) (11,6)};

\draw [gray, dashed] plot [smooth, tension=0.7] coordinates { (0,0) (3, 0.8) (6, 2)};
\draw [gray, dashed] plot [smooth, tension=0.7] coordinates {(6, 7) (11, 7)};

\end{scope}

\begin{scope}%[]

\draw  plot [smooth, tension=0.7] coordinates { (0,0) (3, 0.8) (6, 2)};
\draw  plot [smooth, tension=0.8] coordinates {(6, 2) (9, 4) (11,6)};
\end{scope}

\draw (0, -0.8) node {$0$};
%\draw (-0.8, 0) node {$0$};
\draw (6, -0.2) -- (6, 0.2);
\draw (6, -0.8) node {$\alloc^*_\wealth$};

\draw (11, -0.2) -- (11, 0.2);
\draw (11, -0.8) node {$1$};

\draw (-0.2, 7) -- (0.2, 7);
\draw (-0.7, 7) node {$\infty$};

\draw [gray, dotted] (6, 0) -- (6, 7);
\draw [gray, dotted] (0, 7) -- (6, 7);

\draw (11, 5) node {$\APF_{\wealth}$};
\draw (6.7, 6.3) %[red] 
node {$\APFS_{\wealth}$};

\end{tikzpicture}
\end{minipage}
\begin{minipage}[t]{0.57\textwidth}
\centering
\begin{tikzpicture}[scale = 0.5]

\draw (-0.2,0) -- (11.5, 0);
\draw (0, -0.2) -- (0, 6.2);

\draw (5.7,2) -- (16.2, 2);
\draw (6, 1.7) -- (6, 8.2);

\begin{scope}[ultra  thick]

% \draw [dashed] plot [smooth, tension=0.7] coordinates { (0,0) (3, 0.8) (6, 2)};
% \draw [dashed] plot [smooth, tension=0.8] coordinates {(6, 2) (9, 4) (11,6)};

\draw [gray, dashed] plot [smooth, tension=0.8] coordinates {(6, 2) (9, 4) (11,6)};
\draw [gray, dashed] plot [smooth, tension=0.8] coordinates {(11,8) (16, 8)};

\end{scope}

\draw  plot [smooth, tension=0.7] coordinates { (0,0) (3, 0.8) (6, 2)};
\draw  plot [smooth, tension=0.8] coordinates {(6, 2) (9, 4) (11,6)};

\draw [gray, dotted] (11, 0) -- (11, 8);
\draw [gray, dotted] (6, 0) -- (6, 2);
\draw [gray, dotted] (6, 8) -- (11, 8);

\draw (0, -0.8) node {$0$};

\draw (6, -0.2) -- (6, 0.2);
\draw (6, -0.8) node {$\alloc^*_\wealth$};

\draw (11, -0.2) -- (11, 0.2);
\draw (11, -0.8) node {$1$};

\draw (5.8, 8) -- (6.2, 8);
\draw (5.3, 8) node {$\infty$};

\draw (16, 1.8) -- (16, 2.2);
\draw (16, 1.2) node {$1$};

\draw %[green] 
(9, 5) node {$\APFL_{\wealth}$};
\draw (4.3, 2) node {$\APF_{\wealth}$};

\end{tikzpicture}
\end{minipage}
\end{flushleft}
\caption{\label{f:decompose} 
Depicted are allocation-payment function decomposition.
The black lines in both figures are the allocation-payment function 
$\APF_\wealth$ in ex ante optimal mechanism $\EX$; 
the gray dashed lines are the allocation-payment function
$\APFS_\wealth$ 
and
$\APFL_\wealth$ 
in $\EXS$ and $\EXL$, respectively.
}
\end{figure}
In this construction, 
since the original allocation-payment functions
$\APF_\wealth$ in $\EX$
are convex and non-decreasing
for each budget level $\wealth$,
the constructed
allocation-payment functions 
$\APFS_\wealth$ and $\APFL_\wealth$
are also 
convex and non-decreasing.
Hence, for each type,
the payment in $\EX$
is upper bounded by 
the sum of 
the payments in $\EXS$ and $\EXL$,
and 
the requirements above are satisfied.

\begin{lemma}
\label{lem:EXS}
For a single agent with random-public-budget utility, independently
distributed value and budget, and any ex ante constraint $\quant$; the
revenue of $\EXS$ is at most the revenue from posting the market
clearing price, i.e., $\cumprice(\quant) \geq \Rev{\EXS}$.
\end{lemma}
\begin{proof}
The ex ante allocation of $\EXS$
is at most the ex ante allocation of $\EX$, 
i.e., $\quant$.
Combining with the fact that the per-unit prices
in $\EXS$ for all types are weakly lower 
than the market clearing price,
its revenue is at most the revenue 
of posting the market clearing price.
\end{proof}

Now, we introduce two lemmas 
\Cref{lem:market clearing price large}
and 
\Cref{lem:market clearing price small}
to upper bound the revenue of 
$\EXL$; the former for the case that the market clearing price is large, and the latter for the case that it is small.
At a high level, 
when the market clearing price is large,
we utilize the 
geometry of revenue curves to bound the revenue;
and 
when the market clearing price is small,
we use the 
assumption on the budget distribution
to bound the revenue.

\begin{lemma}[Large Market Clearing Prices]\label{lem:market clearing price large}
For a single 
% item 
% environment with 
% an
agent with random-public-budget utility,
independently distributed value and budget,
and a regular valuation distribution:
\begin{enumerate}
    \item[(i)] for any ex ante 
    constraint $\quant$
    and market clearing price $\marketclearing$
    which is at least the monopoly reserve
    $\myersonReserve$,
    i.e.,
    $\marketclearing = 
    \cumprice(\quant)/\quant 
    \geq \myersonReserve$,
    the revenue of posting the market
    clearing price 
    is 
    at least the revenue of $\EXL$, 
    i.e., $\cumprice(\quant) \geq \Rev{\EXL}$;
    \item[(ii)] (otherwise)
    for any ex ante 
    constraint $\quant$
    and market clearing price 
    $\marketclearing = 
    \cumprice(\quant)/\quant 
    \geq V(\theta)$ for some $\theta$,
    the revenue of posting the market
    clearing price is an $(\sfrac{1}{(1-\theta)})$-approximation 
    to the revenue of $\EXL$, 
    i.e.,
    $\cumprice(\quant) \geq
    (1-\theta)\cdot \Rev{\EXL}$.
\end{enumerate}
%under \Cref{asp:value budget independent},
%\ref{asp:value regular}, 
% for any ex ante constraint~$\quant$,
% the revenue of posting the market clearing 
% price $\price$
% is an $(\sfrac{1}{(1-\theta)})$-approximation 
% to the revenue of $\EXL$, 
% i.e.,
% $\cumprice(\quant) \geq
% (1-\theta)\cdot \Rev{\EXL}
% $,
% if the market clearing price $\price$
% is at least $V(\theta)$;
% further,
% %the revenue of posting the market clearing 
% %price
% the former
% is at least the 
% latter,
% %revenue of $\EXL$,
% i.e.,
% $\cumprice(\quant) \geq \Rev{\EXL}
% $,
% if the market clearing price $\price$
% is at least the Myerson reserve $\myersonReserve$.
\end{lemma}

Note that the third case,
where the market clearing price is small,
will be handled subsequently by \Cref{lem:market clearing price small}. 

\begin{proof}[Proof of \Cref{lem:market clearing price large}]
In both $\EXL$ and  
the mechanism that posts the market clearing price,
the types with value lower than
the market clearing price 
$\marketclearing$ will purchase nothing,
so we only consider the types with value at least $\marketclearing$
in this proof. 
Each budget level is considered separately. 
If a $\beta$-approximation is shown separately, 
then $\beta$-approximation holds in combination. 

For types with budget $\wealth \leq \marketclearing$, 
by posting the market clearing price $\marketclearing$, 
those types always pay their budgets $\wealth$,
which is at least the revenue 
from those types in $\EXL$. 

For types with budget $\wealth > \marketclearing$, 
by posting the market clearing price $\marketclearing$, 
those types always pay $\marketclearing$.
Since the budget constraints do
not bind for these types,
% by 
% posting any price at most
% $\price$,
it is helpful to consider 
the price-posting revenue curve without budget,
which we denote by $\cumpricelinear$.
The regularity of the valuation distribution 
guarantees that
 $\cumpricelinear$ is concave. 
Now we consider two cases: (i)
the market clearing price $\marketclearing$ is 
at least the monopoly reserve
$\myersonReserve$;
(ii)
the market clearing price $\marketclearing$ is 
less than the monopoly reserve
$\myersonReserve$ but at least
$V(\theta)$.
In both cases, we use the concavity of 
$\cumpricelinear$.

% If the market clearing price 
% $\price$ is at least the
% Myerson reserve $\myersonReserve$,
For case (i),
the concavity of $\cumpricelinear$
implies that higher prices above $\myersonReserve$ 
extracts lower revenue
than $\marketclearing$.
Since the per-unit prices in $\EXL$ for all types 
are at least $\marketclearing$,
the concavity of $\cumpricelinear$
guarantees that the expected revenue of posting $\marketclearing$ 
for types with budget larger than
$p$
is at least the expected revenue for those types in $\EXL$.

For case (ii),
% If the $\price$ is between 
% $\myersonReserve$ and $V(\theta)$,
% the revenue from posting $\price$ is
% at least the revenue from 
% posting $V(\theta)$. 
the concavity of $\cumpricelinear$
implies that lower prices below $\myersonReserve$ 
extracts lower revenue.
Thus, posting $\marketclearing$ 
extracts higher revenue than posting  
$V(\theta)$.
Again, the concavity of $\cumpricelinear$
guarantees that 
the revenue of posting  $V(\theta)$ 
is an $(\sfrac{1}{(1-\theta)})$-approximation
to the optimal revenue (i.e., posting 
the monopoly reserve $\myersonReserve$)
generated from these types, 
even with the relaxation of ignoring their budget constraint.

% Therefore, the revenue from posting 
% the market clearing price $\price$ is an $(\sfrac{1}{(1-\theta)})$-approximation
% to $\Rev{\EXL}$ in general, and 
% $2$-approximation if the market clearing price 
% $\price$ is at least 
% the Myerson reserve $\myersonReserve$.

% By posting the market clearing price $\price$, 
% %the market clearing price,
% the types with budget $\wealth \leq \price$ 
% always pay their budget $\wealth$,
% which is at least the revenue from $\EXL$;
% and 
% the types with budget $\wealth > \price$
% always pay $\price$.
% By the assumptions that $\price \geq V(\theta)$
% and valuation distribution is regular, 
% the concavity of the price-posting revenue curve
% guarantees that the revenue from posting 
% the market clearing price is an 
%$(\sfrac{1}{(1-\theta)})$-approximation
% to the optimal revenue without budget constraint.

Combining these bounds above, 
if the market clearing price $\marketclearing$
is at least the monopoly reserve
$\myersonReserve$,
then
$\cumprice(\quant)
\geq
\Rev{\EXL}$;
and otherwise
if the market clearing price $\marketclearing$
is at least $V(\theta)$, then
$\cumprice(\quant)
\geq (1 - \theta)\cdot \Rev{\EXL}
$.
\end{proof}

% In the proof of \Cref{lem:market clearing price large},
% we use the assumption 
% that the market clearing price $\price$
% is at least $V(\theta)$ to upper bound the 
% revenue from the types with budget 
% $\wealth$ larger than the 
% market clearing price $\price$.
% By replacing this assumption with a strong assumption
% that %the market clearing price 
% $\price$
% is at least the Myerson reserve $\myersonReserve$,
% a strong result hold with the same argument.

% \begin{corollary}
% \label{cor:market clearing price large Myerson}
% For a single 
% % item 
% % environment with 
% % an
% agent with private-budget utility,
% assuming the value and budget are
% independently distributed,
% and the valuation distribution is regular,
% %under \Cref{asp:value budget independent},
% %\ref{asp:value regular}, 
% for any ex ante constraint $\quant$,
% if the market clearing price $\price$
% is at least the Myerson reserve $\myersonReserve$,
% then the price-posting revenue 
% is a 
% $2$-approximation 
% to
% the ex ante revenue 
% at ex ante constraint $\quant$, i.e., 
% $2\cumprice(\quant) \geq \revcurve(\quant)$.
% \end{corollary}

\begin{lemma}[Small Market Clearing Prices]\label{lem:market clearing price small}
For a single
% item 
% environment with an 
agent with random-public-budget utility,
independently
distributed
 value and budget,
and a budget distribution 
such that 
the budget exceeds its expectation
with probability at least 
$\sfrac{1}{\budgetQuantile}$,
%under \Cref{asp:value budget independent},
%\ref{asp:budget MHR},
for any ex ante constraint $\quant \leq \sfrac{\theta}{\budgetQuantile}$
with market clearing price $\marketclearing = \cumprice(\quant)/\quant < V(\theta)$,
the revenue of posting 
the market clearing price $\marketclearing$ 
is a 
$(1+\budgetQuantile -
\sfrac{1}{\budgetQuantile})$-approximation to
the revenue of $\EXL$, i.e., 
$(1+\budgetQuantile
-\sfrac{1}{\budgetQuantile}
)\cdot\cumprice(\quant)
\geq \Rev{\EXL}$
\end{lemma}

Before the proof of \Cref{lem:market clearing price small},
we introduce an intermediate lemma used in our argument.

\begin{lemma}\label{lem:decreasing budget}
Fix any convex and 
non-decreasing allocation-payment function,
and any value $\val$,
for any budgets $\wealth\primed$ and
$\wealth\doubleprimed$ such
that $\wealth\primed \leq \wealth\doubleprimed$,
subject to this allocation-payment function,
the payment from the type $(\val, \wealth\primed)$ 
with value $\val$ and budget 
$\wealth\primed$ 
is a 
$(\sfrac{\wealth\doubleprimed}{\wealth\primed})$-approximation
to the payment from the type $(\val, \wealth\doubleprimed)$
with value $\val$ and
budget $\wealth\doubleprimed$.
\end{lemma}
\begin{proof}
% By \Cref{lem:pricing function characterization},
% subject to an allocation-payment function,
% a type will purchase as much as she wants 
% until the budget constraint binds, 
% or the unit-supply constraint binds, or 
% the value binds (i.e., her marginal utility
% becomes zero).

If the type $(\val, \wealth\primed)$ pays
her budget $\wealth\primed$ 
(i.e., the budget constraint binds), 
her payment is a 
$(\sfrac{\wealth\doubleprimed}{\wealth\primed})$-approximation
to the payment from the type $(\val, \wealth\doubleprimed)$,
since the type $(\val, \wealth\doubleprimed)$
pays at most $\wealth\doubleprimed$.

If the type $(\val, \wealth\primed)$ pays less 
than her budget $\wealth$  
(i.e., the unit-demand constraint binds, or 
the value binds),
her allocation is equal to the allocation 
from the type $(\val, \wealth\doubleprimed)$.
Hence, their payments are the same.
\end{proof}

\begin{proof}[Proof of \Cref{lem:market clearing price small}]
Similar to the proof of 
\Cref{lem:market clearing price large},
we only consider the types with value at least $\marketclearing$.
Let $\wealth^*$ be the expected budget. 
We claim that all types with 
budget $\wealth^*$ pay their budget
in posting the market clearing price.
Otherwise, 
the unit-demand constraint binds for all
types $(\val,\wealth)$
such that 
$\val \geq \marketclearing$,
$\wealth \geq \wealth^*$;
and the assumption that $\marketclearing < V(\theta)$
implies
the ex ante allocation for them 
%with budget at least $\wealth^*$
is 
$\prob{\val\geq \marketclearing}\cdot\prob{\wealth \geq \wealth^*}
> 
\sfrac{\theta}{\budgetQuantile}$, 
which exceeds the ex ante constraint~$q$, 
a contradiction.

Let $\Rev[\wealth']{\APFL_\wealth}$
be the expected 
revenue of providing the allocation-payment 
function $\APFL_\wealth$ 
in $\EXL$ 
to the types with budget $\wealth'$;
and let $\Rev[\wealth']{\marketclearing}$ 
be the expected 
revenue of posting 
the market clearing
price $\marketclearing$ 
to the types with budget $\wealth'$.

The following three facts allow 
comparison of
$\Rev{\EXL}$ to 
$\cumprice(\quant)$:
        (a) Posting the market clearing 
        price $\marketclearing$
        makes the budget constraints bind 
        for the types
        with budget at most $\wealth^*$, so
        $\RevAPFww \leq \RevPPw$
        for all $\wealth \leq \wealth^*$.
        (b) 
        \Cref{lem:decreasing budget} implies that 
        $\RevAPFww \leq 
        \frac{\wealth}{\wealth^*}
        \RevAPFwEw$
        for all $\wealth \geq \wealth^*$.
        % (c) 
        % Since $\RevPPw$ is weakly increasing in $\wealth$, then
        % $\int_{\lbudget}^{\wealth^*}\RevPPw dG(\wealth)
        % =
        % \cumprice(q) - \int_{\wealth^*}^{\hbudget}\RevPPw dG(\wealth)
        % \leq 
        % \cumprice(q) - \cumprice(q) \cdot \prob{\wealth\geq \wealth^*}
        % =
        % (1-\sfrac{1}{\budgetQuantile})\cumprice(q)$.
        (c) Since 
        the revenue of posting 
        the market clearing price $\marketclearing$ 
        to an agent with budget $\wealth^*$
        is at most 
        the revenue
        to an agent with budget $\wealth > \wealth^*$;
        with the assumption that
        budgets exceed the 
        expectation $\wealth^*$ 
        with probability at least 
        $\sfrac{1}{\budgetQuantile}$, 
        it implies that 
        $$
        \RevPPEw\cdot \frac{1}{\budgetQuantile}
        \leq 
        \expect{\RevPPw
        \ \Big| \ 
        \wealth \geq \wealth^*}\cdot
        \prob{\wealth\geq \wealth^*}
        \leq \cumprice(\quant).
        %\qedhere
        $$
We upper bound the revenue of $\EXL$ as follows,
\begin{align*}
        \Rev{\EXL} &=
        \int_{\lbudget}^{\wealth^*} 
        \RevAPFww dG(\wealth) 
        +
        \int_{\wealth^*}^{\hbudget} 
        \RevAPFww dG(\wealth) \\
        &\leq 
        \int_{\lbudget}^{\wealth^*} 
        \RevPPw dG(\wealth)
        +
        \int_{\wealth^*}^{\hbudget} 
        \frac{\wealth}{\wealth^*} 
        \RevAPFwEw dG(\wealth) \\
        &\leq
        (1-\frac{1}{\budgetQuantile})
        \cumprice(\quant)
        +
        \frac{\int_{\wealth^*}^{\hbudget}
        \wealth dG(\wealth)}{\wealth^*} 
        \RevPPEw  \\
        &\leq
        (1-\frac{1}{\budgetQuantile})
        \cumprice(\quant)
        +
        \RevPPEw \leq
        (1+\budgetQuantile-
        \frac{1}{\budgetQuantile})
        \cumprice(\quant)
        \end{align*}
where the first inequality is due to facts (a) and (b); 
in the second inequality, the first term is due to 
$\prob{\wealth \leq \wealth^*} 
\leq 1 - \sfrac{1}{\budgetQuantile}$, 
the revenue $\RevPPw$ is monotone increasing in $\wealth$, 
and by definition $\int_{\lbudget}^{\hbudget} 
        \RevPPw dG(\wealth) = \cumprice(\quant)$, 
%$\RevPPw = \cumprice(\quant)$, 
and the second term is due to fact (a);
and the last inequality is due to fact (c).
\end{proof}

We combine the three lemmas into a proof for
\Cref{thm:private bound}.
\begin{proof}[Proof of \Cref{thm:private bound}]
Set the parameter $\theta = 
\sfrac{\budgetQuantile}{(\budgetQuantile+1)}
% \sfrac{(\budgetQuantile^2-1)}
% {(\budgetQuantile^2 + \budgetQuantile - 1)}
$.
\Cref{lem:EXS}, \Cref{lem:market clearing price large} and 
\Cref{lem:market clearing price small}
implies that for all 
ex ante constraints $\quant \leq
\sfrac{1}{(\budgetQuantile+1)}
% \tfrac{\budgetQuantile^2-1}
% {(\budgetQuantile^2 + \budgetQuantile -
% 1)\budgetQuantile}
$,
the price posting revenue is a 
$(2+\budgetQuantile
%-\sfrac{1}{\budgetQuantile}
)$-approximation 
to the ex ante revenue.
% Hence, applying \Cref{lem:close rev curve imply close pricing},
% % and \Cref{thm:general AP bound},
% % the $(\alphaprivate,\betaprivate)$-close property
% % holds and the approximation 
% % bound for anonymous pricing 
% the result
% is proved.
\end{proof}

% \begin{remark}
% Even though in \Cref{thm:private bound}, 
% we assume that the budget distribution is MHR,
% we only use the fact that the quantile of its 
% expectation is at least $\sfrac{1}{e}$ in our proof.
% A similar constant approximation result follows, 
% by replacing this assumption with any constant 
% lower bound guarantee on the expected budget.
% \end{remark}

\subsection{MHR Budget Distributions}

Assuming that
the value and the budget 
are independently distributed,
the valuation distribution is regular,
and
the budget distribution 
is MHR, 
the price-posting revenue curve is concave. 
This result is formally stated in 
\Cref{lem:private budget concave price-posting revenue curve}, 
whose proof is deferred to 
Appendix~\ref{apx:private}. 
Due to the concavity of the price-posting revenue curve, 
we show
a better approximation bound 
for anonymous pricing 
using the reduction in \Cref{thm:general EAR bound}. 
% by
% \Cref{thm:general EAR bound}.
% The proof of \Cref{lem:private budget concave price-posting revenue curve}
% %and \Cref{thm: private budget MHR}
% is deferred to 
% Appendix~\ref{apx:private}.

\begin{lemma}
\label{lem:private budget concave price-posting revenue curve} 
A single private-budget agent has a concave price-posting
revenue curve $\cumprice$ if her value and budget are independently
distributed, the valuation distribution is regular, and the budget
distribution is MHR.
\end{lemma}

% With \Cref{lem:private budget concave price-posting revenue curve}, 
% the price-posting revenue is concave, 
% and we are able to achieve better bounds for 
% MHR budget distributions by applying the approach 
% in \Cref{cor:ex ante close}. 

\begin{theorem}\label{thm:private budget MHR}
A private-budget agent is $\mhrbound$-close for ex ante optimization
if her value and budget are independently distributed, the valuation
distribution is regular, and the budget distribution is MHR.
\end{theorem}

% \end{theorem}

\begin{corollary}
For a single item environment with private-budget agents, anonymous
pricing is a $\mhrbound\rho$-approximation to the ex ante relaxation
if values and budgets are independently distributed, the valuation
distributions are regular, and the budget distributions are MHR.
\end{corollary}

When the market clearing price is larger than the monopoly reserve,
\Cref{lem:EXS} and \Cref{lem:market clearing price large} guarantees
that posting the market clearing price is a 2-approximation to $\EX$.
In this subsection, we improve the approximation guarantee in other
case where the market clearing price is smaller than the monopoly
reserve.  This improvement is due to the fact that closeness for ex
ante optimization is a weaker condition than closeness for price
posting, which allows us to consider posting any price at least the
market clearing price, i.e., the ex ante constraint does not need to
bind.  Specifically, we will use the technique from \Cref{lem:single
  agent approx public} which considers a random per-unit price drawn
from the valuation distribution; however, if the realization of the
random price is smaller than the the market clearing price, to satisfy
the ex ante constraint, we replace it with the market clearing price.

The proof for small market clearing prices based on the decomposition
from the previous subsection.  Fix any ex ante constraint $\quant$, we
decompose the ex ante mechanism $\EX$ into $\EXS$ and $\EXL$ where the
per-unit prices in $\EXS$ for all types are at most the market
clearing price and the per-unit prices in $\EXL$ for all types are
larger than the market clearing price; and then we upper bound the
revenue of $\EXS$ and $\EXL$.  For $\EXS$ the previously proved bounds
are sufficient; the following lemma improves the bound for $\EXL$.

\begin{figure}[t]
\begin{flushleft}
\hspace{-5pt}
\begin{minipage}[t]{0.48\textwidth}
\centering
\begin{tikzpicture}[scale = 0.45]

% \draw [name path = E, white] (6, 4.714285714285714) -- (0, 0);

% \draw [name path global = B, white]
% (6, 4.714285714285714) .. controls 
% (11, 3.75) .. (12, 0);

\draw [name path = A, white, dotted] (9.1, 0) -- (9.1, 2.654);
\draw [name path = D, white, dotted] (0, 0) -- (9.1, 2.654);

\draw [name path = C, dashed, thick] (0, 0) -- (12, 5.3);
\draw [dotted, thick] (0, 0) -- (12, 3.5);

\tikzfillbetween[of=A and D]{gray!40!white};
\fill[color=gray!40!white] (9.1, 2.654) -- (12, 2.654)
-- (12, 0) -- (9.1, 0);

\draw [gray, dotted] (9.1, 0) -- (9.1, 4.05);

%\tikzfillbetween[of=B and E]{gray!40!white};

\fill[pattern=north west lines, pattern color=gray!20!white] (0, 0) -- (0, 2.654) 
-- (12, 2.654) -- (12, 0);

\draw (-0.2,0) -- (12.5, 0);
\draw (0, -0.2) -- (0, 6);

\begin{scope}[very thick]

\draw [name path global = A](0, 0) .. controls 
(2, 5.5) .. (6, 4.714285714285714);

\draw [name path global = F]
(6, 4.714285714285714) .. controls 
(11, 3.75) .. (12, 0);

\end{scope}

\draw (0, -0.8) node {$0$};
\draw (12, -0.8) node {$1$};

\draw (10.7, 5.3) node {$\marketclearing$};
\draw (12, 4) node {$\wealth$};

\draw [gray, dotted] (0, 2.654) -- (12, 2.654);
\draw [gray, dotted] (12, 0) -- (12, 2.654);
\draw (-2.1, 2.654) node {$\RevPPw$};

\draw (-0.2, 2.654) -- (0.2, 2.654);

\end{tikzpicture}
\end{minipage}
\begin{minipage}[t]{0.48\textwidth}
\centering
\begin{tikzpicture}[scale = 0.45]

\draw [name path = E, white] (6, 4.714285714285714) -- (0, 0);

\draw [name path global = B, white]
(6, 4.714285714285714) .. controls 
(11, 3.75) .. (12, 0);

\fill[color=gray!40!white] (0, 0) -- (6, 4.714285714285714) 
-- (12, 0);
\draw[color=gray!40!white]
(6, 4.714285714285714) -- (12,0);
\fill[color=gray!40!white] (9, 4.05) -- (12, 4.05)
-- (12, 0) -- (9, 0);

\tikzfillbetween[of=B and E]{gray!40!white};

\fill[pattern=north west lines, pattern color=gray!20!white] (0, 0) -- (0, 4.714285714285714) 
-- (12, 4.714285714285714) -- (12, 0);

\draw (-0.2,0) -- (12.5, 0);
\draw (0, -0.2) -- (0, 6);

\begin{scope}[very thick]

% \draw [dotted] (6.7, 0) -- (6.7, 7.3);
% \draw [dotted] (4.8, 0) -- (4.8, 4.8);
% \draw [dotted] (3.4, 0) -- (3.4, 6.8);

%\draw plot [smooth, tension=0.6] coordinates { (0,0) (1,3.6) (4, 5) (8,4) (12,0)};

\draw [name path global = A](0, 0) .. controls 
(2, 5.5) .. (6, 4.714285714285714);

\draw [name path global = F]
(6, 4.714285714285714) .. controls 
(11, 3.75) .. (12, 0);

\end{scope}

\draw [name path = C, dashed, thick] (0, 0) -- (12, 5.3);
\draw [name path = D, dotted, thick] (0, 0) -- (7, 5.5);

\draw (0, -0.8) node {$0$};
\draw (12, -0.8) node {$1$};

\draw (10.7, 5.3) node {$\marketclearing$};
\draw (7.5, 5.5) node {$\wealth$};

\draw [gray, dotted] (0, 4.714285714285714) -- (12, 4.714285714285714);
\draw [gray, dotted] (12, 0) -- (12, 4.714285714285714);
\draw (-2.99, 4.8) node
{$\Rev[\wealth]{\OPT_\wealth}$};

\draw [gray, dotted] (9, 4.05) -- (12, 4.05);
\draw [gray, dotted] (12, 0) -- (12, 4.05);

\draw (-0.2, 4.714285714285714) -- (0.2, 4.714285714285714);

\end{tikzpicture}

% \begin{tikzpicture}[scale = 0.5]

% \draw (-0.2,0) -- (12.5, 0);
% \draw (0, -0.2) -- (0, 6);

% \begin{scope}[very thick]

% % \draw [dotted] (6.7, 0) -- (6.7, 7.3);
% % \draw [dotted] (4.8, 0) -- (4.8, 4.8);
% % \draw [dotted] (3.4, 0) -- (3.4, 6.8);

% \draw plot [smooth, tension=0.6] coordinates { (0,0) (1,3.6) (4, 5) (8,4) (12,0)};

% % \draw [dashed] plot [smooth, tension=0.6] coordinates { (0,0) (0.8,2) (4, 3.5) (7.5, 1.5) (9,1.1) (12,0)};
% \draw (0, 0) -- (12, 3);
% \draw (0, 0) -- (7, 5.5);
% % \draw (0, 0) -- (4.5, 9);

% \end{scope}

% \draw (0, -0.8) node {$0$};
% \draw (12, -0.8) node {$1$};

% \draw (12, 3.5) node {$\price$};
% \draw (7.5, 5.5) node {$\wealth$};

% \draw [gray, dotted] (0, 4.8) -- (6, 4.8);
% \draw (-2.5, 4.8) node {$\Rev[\wealth]{\OPT}$};

% \draw [gray, dotted] (9.7, 2.4) -- (12, 2.4);
% \draw [gray, dotted] (12, 0) -- (12, 2.4);

% \end{tikzpicture}
\end{minipage}
\end{flushleft}
\caption{\label{f:revenue of sample} In the geometric proof of
  \Cref{lem:posting randomized price for EXL}, the upper bound on the
  expected revenue of $\EXL$ ($\RevPPw$ and
  $\Rev[\wealth]{\OPT_\wealth}$ on the left and right, respectively)
  is the area of the light gray striped rectangle and the revenue from
  posting random price $\randomprice$ is the area of the dark gray
  region.  By geometry, the latter is at least half of the former.
  The black curve is the price-posting revenue curve with no budget
  constraint $\cumpricelinear$.  The figure on the left depicts the
  small-budget case (i.e., $\wealth < \marketclearing$), and the
  figure on the right depicts the large-budget case (i.e., $\wealth
  \geq \marketclearing$).  }
\end{figure}

\begin{lemma}
\label{lem:posting randomized price for EXL}
For a single private-budget agent with independently distributed value
and budget and regular value distribution, if the market clearing
price $\marketclearing = \sfrac{\cumprice(\quant)}{\quant}$ is smaller
than the monopoly reserve, there exists $\quant\primed \leq \quant$
such that the market clearing revenue from $\quant\primed$ is a
2-approximation to the revenue from $\EXL$, i.e., $2
\cumprice(\quant\primed) \geq \Rev{\EXL} $.
\end{lemma}

\begin{proof}
Note that any price that is at least $\marketclearing$ is feasible for
the ex ante constraint $\quant$.  We consider posting a random price
$\randomprice = \max\{\marketclearing, \randomprice_0\}$ with
$\randomprice_0$ drawn identically to the agents value distribution.
Fixing the budget of the agent $\wealth$, consider the following
geometric argument \citep[cf.][]{DRY-10}.  For both sides of
\Cref{f:revenue of sample}, the area of the light gray stripped
rectangle upper bounds the revenue of $\EXL$ and the area of the dark
gray region is the expected revenue from posting random price
$\randomprice$.  Consequently, concavity of the price-posting revenue
curve with no budget constraint $\cumpricelinear$ (by regularity of
the value distribution) implies that a triangle with half the area of
the light gray rectangle is contained within the dark gray region and,
thus, the random price is a 2-approximation.  As the random price does
not depend on the budget $\wealth$, the same bound holds when
$\wealth$ is random.  Of course, the optimal deterministic price that
is at least $\marketclearing$ is only better than the random price and
the lemma is shown.  The remainder of this proof verifies that the
geometry of the regions described above is correct.

The left side of \Cref{f:revenue of sample} depicts the fixed budgets
$\wealth$ that are at most $\marketclearing$.  The area of the light
gray striped rectangle upper bounds the revenue of $\EXL$ as follows.
Let $\Rev[\wealth]{\price}$ be the expected revenue from posting price
$\price$ to types with budget $\wealth$.  Under both $\EXL$ and the
market clearing price $\marketclearing$, types with value below the
market clearing price pay zero.  For the remaining types, in $\EXL$
they pay at most their budget and in market clearing they pay exactly
their budget.  Thus, $\Rev[\wealth]{\EXL} \leq
\Rev[\wealth]{\marketclearing} = \wealth \, (1-F(\marketclearing))$
where, recall, $1-F(\marketclearing)$ is the probability the agent's
value is at least the market clearing price $\marketclearing$.  Of
course, $\wealth \, (1-F(\marketclearing))$ is the height and area (its
width is 1) of the light gray striped region on the left side of
\Cref{f:revenue of sample}.

The right side of \Cref{f:revenue of sample} depicts the fixed budgets
$\wealth$ that are at least $\marketclearing$.  The area of the light
gray striped rectangle upper bounds the revenue of $\EXL$ as follows.
Let $\OPT_\wealth$ be the optimal mechanism to types with budget
$\wealth$ without ex ante constraint and $\Rev[\wealth]{\OPT_\wealth}$
be its expected revenue from these types.  Clearly,
$\Rev[\wealth]{\EXL} \leq \Rev[\wealth]{\OPT_\wealth}$ as the latter
optimizes with relaxed constraints of the former.  \Cref{lem:public
  budget optimal under regular} implies that $\OPT_\wealth$ posts the
minimum between budget $\wealth$ and the monopoly reserve
$\myersonReserve$.  As the budget does not bind for this price, its
revenue is given by the price-posting revenue curve with no budget
constraint, i.e., $\Rev[\wealth]{\OPT_\wealth} = \cumpricelinear(1 -
F(\min\{\wealth, \myersonReserve\}))$.  Of course, this revenue is the
height and area (its width is 1) of the light gray striped region on
the right side of \Cref{f:revenue of sample}.

Next, we will show that the revenue of posting the random price
$\randomprice$ is the grey shaded areas illustrated in \Cref{f:revenue
  of sample} (in both cases).  A random price from the value
distribution, i.e.,  $\randomprice_0$, corresponds to a uniform random quantile
constraint, i.e., drawing uniformly from the horizontal axis.  Since
we truncate the lower end of the price distribution at the market
clearing price $\marketclearing$, the revenue from quantiles greater
than $\quant$ equals the revenue from the market clearing price.  For
any fixed $\wealth$, when $\randomprice \in [\marketclearing,
  \wealth]$, the budget does not bind and the revenue of posting price
$\randomprice$ is $\cumpricelinear(\randomquant)$ where
$\cumpricelinear$ is the price-posting revenue curve without budget;
and when $\randomprice > \wealth$, the revenue of posting
price~$\randomprice$ is~$\wealth \randomquant$.  Thus, the revenue
from a random price is given by the integral of the area under the
curve defined by $\randomquant \wealth$ when $\randomprice \geq
\wealth$, by $\cumpricelinear(\randomquant)$ when $\randomprice \in
       [\wealth, \marketclearing]$ and this interval exists, and by
       $\min(\wealth,\marketclearing)$ when $\randomprice =
       \marketclearing$, i.e., when $\randomprice_0 \leq
       \marketclearing$.  This area is the dark gray region. \qedhere

\end{proof}

\begin{proof}
[Proof of \Cref{thm:private budget MHR}]
Fix any ex ante constraint $\quant$.
If the market clearing price $\marketclearing = 
\sfrac{\cumprice(\quant)}{\quant}$
is at least the monopoly reserve,
\Cref{lem:EXS} and \Cref{lem:market clearing price large}
imply that
$\Rev{\EXS} \leq \cumprice(\quant)$,
and
$\Rev{\EXL} \leq \cumprice(\quant)$,
thus,
$\cumprice(\quant)$ 
%posting the market clearing price 
%$\price$ 
is a 2-approximation to
$\Rev{\EXS} + \Rev{\EXS} = \Rev{\EX}$, 
i.e., $\revcurve(\quant)$.
If the market clearing price $\marketclearing$
is smaller than the monopoly reserve,
let 
$\quant\primed
=\argmax_{\quant' \leq \quant}
\cumprice(\quant')$,
\Cref{lem:EXS} and 
\Cref{lem:posting randomized price for EXL}
imply 
that 
$\Rev{\EXS} \leq \cumprice(\quant) \leq 
\cumprice(\quant\primed)$,
and
$\Rev{\EXL} \leq 2\cumprice(\quant\primed)$,
thus,
$\cumprice(\quant\primed)$
is a 3-approximation to
$\revcurve(\quant)$.
Thus, the agent
is 3-close for ex ante optimization.
\end{proof}

\section{Risk Averse Utility}
\label{sec:risk-averse}
In this section, 
we consider the case when agents are risk averse. 
Specifically, we consider the risk aversion model 
in \cite{FHH-13}, 
where each agent's utility function has a capacity constraint. 
% Given allocation $\alloc$ and 
% payment $\price$,
% agent's utility with capacity constraint $\capacity$ 
% is $\min \{\val\alloc - \price, \capacity\}$. 
% We refer to this utility function as \emph{capacitated utility}. 
Moreover, following \cite{FHH-13}, in this section, we consider the
mechanisms that are pointwise individual rational, i.e., losers have
no payment, and winners pay at most their reported values.  Formally,
$x = 0$ implies $p = 0$.  In \Cref{exp:single risk agent with overpay}
at the end of this section, we show that price-posting mechanism is
not a constant approximation to the optimal mechanism when we allow
the winners to be charged more than their reported values, even when
the capacity is as large as the support of the value.
%The proof of \Cref{thm:risk-averse bound} 
%is deferred to Appendix~\ref{apx:risk averse proof}.
%are ex-post individual rational. 

We introduce a definition and two lemmas, 
which are adapted from \cite{FHH-13}. Let $(\cdot)^+ = \max\{\cdot, 0\}$. 

% \begin{definition}[\cite{FHH-13}]\label{def:two price}
% A mechanism is a two priced mechanism 
% if when it serves an agent $i$ 
% with quantile $\quant_i$, 
% value $\val_i = \frac{\cumprice_i(\quant_i)}{\quant_i}$, 
% and capacity $\capacity_i$, 
% the payment is either $\val_i$ 
% or $\val_i - \capacity_i$. 
% The probability that agent $i$ is charged with 
% payment $\val_i$ is denoted by $\alloc_{i}^{\val}(\quant_i)$, 
% and the probability that agent $i$ is charged with 
% payment $\val_i - \capacity$ 
% is denoted by $\alloc_{i}^{\capacity}(\quant_i)$. 
% \end{definition}
\begin{definition}[\citealp{FHH-13}]\label{def:two price}
A mechanism is a two priced mechanism 
if, when it serves an agent 
with quantile $\quant$ 
%value $\val = \frac{\cumprice(\quant)}{\quant}$, 
and capacity $\capacity$, 
the payment is either $\valfunc(\quant)$ 
or $\valfunc(\quant) - \capacity$. 
The probability that agent is charged with 
payment $\valfunc(\quant)$ is denoted by $\alloc^{\val}(\quant)$, 
and the probability that agent is charged with 
payment $\valfunc(\quant) - \capacity$ 
is denoted by $\alloc^{\capacity}(\quant)$. 
\end{definition}

\begin{lemma}[\citealp{FHH-13}]\label{lem:two price is opt}
The ex ante optimal mechanism for agents
with capacitated utility is two priced. 
\end{lemma}

\begin{lemma}[\citealp{FHH-13}]\label{lem:upper bound for risk averse}
For any agent with capacity $\capacity$ 
and price-posting revenue curve $\cumprice$, 
for two priced mechanism with allocation rule 
$\alloc(\quant) = \alloc^{\val}(\quant) 
+ \alloc^{\capacity}(\quant)$, 
the revenue from that agent is upper bounded as
\begin{equation*}
\Rev{\alloc} \leq 
\expect{(\cumprice'(\quant))^+ \cdot \alloc(\quant)}
+ \expect{(\cumprice'(\quant))^+ 
\cdot \alloc^{\capacity}(\quant)}
+ \expect{\left(\valfunc(\quant)
- \capacity \right)^+ 
\cdot \alloc^{\capacity}(\quant)}.
\end{equation*}
\end{lemma}

\begin{theorem}
\label{thm:risk-averse bound}
For a single agent 
with maximum value $\maxval$ 
and capacity $\capacity \leq \maxval$, % \geq \max_q \cumprice(q)$, 
if her price-posting revenue curve $\cumprice$ is concave,  
she is $(2 + \ln \sfrac{\maxval}{\capacity})$-close 
for ex ante optimization. 
\end{theorem}

\begin{proof}
For any quantile $\optquant$, 
let $\alloc$ be the optimal allocation 
that satisfies ex ante allocation constraint $\optquant$. 
By \Cref{lem:upper bound for risk averse}, 
\begin{equation*}
\revcurve(\optquant) = \Rev{\alloc} \leq 
\expect{(\cumprice'(\quant))^+ \cdot \alloc(\quant)}
+ \expect{(\cumprice'(\quant))^+ 
\cdot \alloc^{\capacity}(\quant)}
+ \expect{\left(\valfunc(\quant) 
- \capacity \right)^+ 
\cdot \alloc^{\capacity}(\quant)}.
\end{equation*}
Let $\myersonReserve$ be the monopoly reserve, 
and let 
$\correspondingquant = \min\{
\priceToQuantile(\myersonReserve, \cumprice), \optquant\}$. 
By definition, $\correspondingquant \leq \optquant$. 
Since the price-posting revenue curve is concave, 
posting price 
$\valfunc(\correspondingquant)$
maximizes expected marginal revenue 
under ex ante constraint $\optquant$. 
Therefore, 
$$
\expect{(\cumprice'(\quant))^+ \cdot \alloc(\quant)} 
\leq \cumprice(\correspondingquant)
$$
and 
$$
\expect{(\cumprice'(\quant))^+ 
\cdot \alloc^{\capacity}(\quant)} 
\leq \cumprice(\correspondingquant). 
$$
When $\correspondingquant 
= \priceToQuantile(\myersonReserve, \cumprice)$, 
for any quantile $\quant$, 
$\cumprice(\quant) \leq \cumprice(\correspondingquant)$. 
When $\correspondingquant = \optquant < 
\priceToQuantile(\myersonReserve, \cumprice)$, 
the allocation $\alloc^{\capacity}(\quant)$ 
with ex ante constraint $\optquant$
that maximizes 
$\expect{\left(\valfunc(\quant) 
- \capacity \right)^+ 
\cdot \alloc^{\capacity}(\quant)}$ 
satisfies that 
$\alloc^{\capacity}(\quant) = 1$ 
for $\quant \leq \correspondingquant$,
and $\alloc^{\capacity}(\quant) = 0$
for $\quant > \correspondingquant$. 
Since the price-posting revenue curve is concave, 
in this case, $\cumprice(\quant) \leq \cumprice(\correspondingquant)$ 
when $\quant \leq \correspondingquant$. 
Therefore, 
% Moreover, since the price-posting revenue curve is concave, 
% $\cumprice(\quant) \leq \cumprice(\correspondingquant)$ 
% when $\quant \leq \correspondingquant$, 
% or when $\quant > \correspondingquant$ 
% and $\correspondingquant 
% = \priceToQuantile(\myersonReserve, \cumprice)$. 
% Thus, 
\begin{eqnarray*}
&&\expect{\left(\valfunc(\quant) 
- \capacity \right)^+ 
\cdot \alloc^{\capacity}(\quant)} 
=
\expect{\left(\frac{\cumprice(\quant)}{\quant} 
- \capacity \right)^+ 
\cdot \alloc^{\capacity}(\quant)} \\
&\leq& \expect{\left(\min \left\{\maxval,\frac{
\cumprice(\correspondingquant)}{\quant} \right\} 
- \capacity \right)^+} \\
&=& \int_{\frac{\cumprice(\correspondingquant)}{\maxval}}
^{\min\{1, 
\frac{\cumprice(\correspondingquant)}{\capacity}\}}
\left(\frac{
\cumprice(\correspondingquant)}{\quant} 
- \capacity \right) d\quant
+ \int_{\frac{\cumprice(\correspondingquant)}{\maxval}}^1 
(\maxval - \capacity) dq \\
&\leq& \cumprice(\correspondingquant) 
\ln \frac{\maxval}{\capacity}. 
\end{eqnarray*}
Combining the above inequalities, we have
$\revcurve(q) \leq 
\cumprice(\correspondingquant) 
(2 + \ln \frac{\maxval}{\capacity})$.
\end{proof}

When the agent's value distribution is regular, 
her price-posting revenue curve is concave. 
By applying \Cref{cor:ex ante close} 
and \Cref{lem:closeness is downward imply}, 
we have the following result. 

\begin{corollary}\label{cor:risk averse}
For a single item 
environment where 
each agent $i$ has capacity $\capacity_i$, 
and regular valuation distributions, 
letting $\mincapacity = \min_i (\capacity_i)$ 
be the minimum capacity, 
the approximation factor of
anonymous pricing to the ex ante
relaxation is at most $(2 + \ln \sfrac{\maxval}{\mincapacity})\rho$. 
\end{corollary}

% In this section, as in \citet{FHH-13}, we require that 
% % losers have no payment, 
% % and winners pay at most their bids. 
% the mechanism is pointwise individual rational. 
% In \Cref{exp:single risk agent with overpay}, we show that 
% price-posting mechanism is not a constant approximation 
% to the optimal mechanism
% when we allow the winners to pay more than their bids, 
% even when the capacity is as large as the support of the value.  
In \Cref{cor:risk averse},
the dependence on 
$\ln \sfrac{\maxval}{\mincapacity}$ is necessary 
even when there is a single agent. 
%This dependence is proved by \Cref{exp:single risk agent}. 

\begin{example}[necessity of the dependence on $\sfrac{\maxval}{\mincapacity}$]\label{exp:single risk agent}
Fix a constant $\maxval$.  Consider a single agent with equal revenue
distribution.  That is, her value $\val$ is drawn from $[1, \maxval]$
with a density function $\sfrac{1}{\val^2}$ for $v \in [1, \maxval)$,
  and a mass point of probability $\sfrac{1}{\maxval}$ on value
  $\maxval$.  The revenue for posting any price is $1$.  Suppose the
  agent has capacity constraint $\capacity \geq 1$, Consider the
  mechanism that always allocates the item to the agent, and charges
  her $0$ if her value $\val$ is less than $\capacity$, and charges
  her $\val - \capacity$ if her value is at least $\capacity$.  The
  revenue for this mechanism is $\ln \sfrac{\maxval}{\capacity}$.
\end{example}

\begin{example}[necessity of the restriction to pointwise individually rational mechanisms]\label{exp:single risk agent with overpay}
Fix a constant $\maxval$.  Consider a single agent with equal revenue
distribution as in \Cref{exp:single risk agent}.  The revenue for
posting any price is $1$.  Suppose the agent has capacity constraint
$\capacity = \maxval$ and consider the mechanism that always allocates
the item to the agent, and charges her $v-\maxval$ with probability
$\frac{1}{2}$, $\maxval$ with probability $\frac{1}{2}$.  This
mechanism is incentive compatible and individually rational.  The
revenue for this mechanism is half of the welfare, which cannot be
approximated within a constant fraction by any price-posting
mechanism.
\end{example}

\section{Conclusions}
\label{sec:conclusion}
In this paper we focused on the approximation bounds for anonymous
pricing in single item environments.  We introduce a generalization of
regularity which characterizes the gap between the price-posting
revenue curve and the ex ante revenue curve.  This generalization
enable a reduction framework to approximately reduce the analysis of
the approximation bound for anonymous pricing for agents with
non-linear utility to that of agents with linear utility.
We believe this generalization can be applied for the analysis 
of a broader class of mechanisms. 
For instance, the revenue guarantee of the sequential posted price 
mechanism is an interesting direction to apply the approach.

\bibliography{Paper/auctions}

\begin{thebibliography}{}

\bibitem[Abrams, 2006]{abr-06}
Abrams, Z. (2006).
\newblock Revenue maximization when bidders have budgets.
\newblock In {\em Proceedings of the seventeenth annual ACM-SIAM symposium on
  Discrete algorithm}, pages 1074--1082.

\bibitem[Alaei, 2011]{ala-11}
Alaei, S. (2011).
\newblock Bayesian combinatorial auctions: Expanding single buyer mechanisms to
  many buyers.
\newblock In {\em Proc. 52th IEEE Symp. on Foundations of Computer Science}.

\bibitem[Alaei et~al., 2013]{AFHH-13}
Alaei, S., Fu, H., Haghpanah, N., and Hartline, J. (2013).
\newblock The simple economics of approximately optimal auctions.
\newblock In {\em Proc. 54th IEEE Symp. on Foundations of Computer Science},
  pages 628--637.

\bibitem[Alaei et~al., 2012]{AFHHM-12}
Alaei, S., Fu, H., Haghpanah, N., Hartline, J., and Malekian, A. (2012).
\newblock Bayesian optimal auctions via multi-to single-agent reduction.
\newblock In {\em Proceedings of the 13th ACM Conference on Electronic
  Commerce}, page~17.

\bibitem[Alaei et~al., 2018]{AHNPY-18}
Alaei, S., Hartline, J., Niazadeh, R., Pountourakis, E., and Yuan, Y. (2018).
\newblock Optimal auctions vs. anonymous pricing.
\newblock {\em Games and Economic Behavior}.

\bibitem[Barlow and Marshall, 1965]{BM-65}
Barlow, R.~E. and Marshall, A.~W. (1965).
\newblock Tables of bounds for distributions with monotone hazard rate.
\newblock {\em Journal of the American Statistical Association},
  60(311):872--890.

\bibitem[Bulow and Roberts, 1989]{BR-89}
Bulow, J. and Roberts, J. (1989).
\newblock The simple economics of optimal auctions.
\newblock {\em The Journal of Political Economy}, 97:1060--90.

\bibitem[Chawla et~al., 2007]{CHK-07}
Chawla, S., Hartline, J., and Kleinberg, R. (2007).
\newblock Algorithmic pricing via virtual valuations.
\newblock In {\em Proc. 8th ACM Conf. on Electronic Commerce}.

\bibitem[Chawla et~al., 2010]{CHMS-10}
Chawla, S., Hartline, J., Malec, D., and Sivan, B. (2010).
\newblock Sequential posted pricing and multi-parameter mechanism design.
\newblock In {\em Proc. 41st ACM Symp. on Theory of Computing}.

\bibitem[Chawla et~al., 2011]{CMM-11}
Chawla, S., Malec, D.~L., and Malekian, A. (2011).
\newblock Bayesian mechanism design for budget-constrained agents.
\newblock In {\em Proceedings of the 12th ACM conference on Electronic
  commerce}, pages 253--262. ACM.

\bibitem[Che and Gale, 2000]{CG-00}
Che, Y.-K. and Gale, I. (2000).
\newblock The optimal mechanism for selling to a budget-constrained buyer.
\newblock {\em Journal of Economic theory}, 92(2):198--233.

\bibitem[Che et~al., 2013]{CKM-13}
Che, Y.-K., Kim, J., and Mierendorff, K. (2013).
\newblock Generalized reduced-form auctions: A network-flow approach.
\newblock {\em Econometrica}, 81(6):2487--2520.

\bibitem[Devanur and Weinberg, 2017]{DW-17}
Devanur, N.~R. and Weinberg, S.~M. (2017).
\newblock The optimal mechanism for selling to a budget constrained buyer: The
  general case.
\newblock In {\em Proceedings of the 2017 ACM Conference on Economics and
  Computation}, pages 39--40. ACM.

\bibitem[Dhangwatnotai et~al., 2010]{DRY-10}
Dhangwatnotai, P., Roughgarden, T., and Yan, Q. (2010).
\newblock Revenue maximization with a single sample.
\newblock In {\em ECOM10}.

\bibitem[Feng and Hartline, 2018]{FH-18}
Feng, Y. and Hartline, J.~D. (2018).
\newblock An end-to-end argument in mechanism design (prior-independent
  auctions for budgeted agents).
\newblock In {\em 2018 IEEE 59th Annual Symposium on Foundations of Computer
  Science (FOCS)}, pages 404--415. IEEE.

\bibitem[Fu et~al., 2013]{FHH-13}
Fu, H., Hartline, J., and Hoy, D. (2013).
\newblock Prior-independent auctions for risk-averse agents.
\newblock In {\em Proceedings of the fourteenth ACM conference on Electronic
  commerce}, pages 471--488.

\bibitem[Hartline and Roughgarden, 2009]{HR-09}
Hartline, J. and Roughgarden, T. (2009).
\newblock Simple versus optimal mechanisms.
\newblock In {\em Proc. 10th ACM Conf. on Electronic Commerce}, pages 225--234.

\bibitem[Holt~Jr, 1980]{Hol-80}
Holt~Jr, C.~A. (1980).
\newblock Competitive bidding for contracts under alternative auction
  procedures.
\newblock {\em Journal of Political Economy}, 88(3):433--445.

\bibitem[Jin et~al., 2019a]{JLQTX-19}
Jin, Y., Lu, P., Qi, Q., Tang, Z.~G., and Xiao, T. (2019a).
\newblock Tight approximation ratio of anonymous pricing.
\newblock In {\em Proceedings of the 51th Annual ACM Symposium on Theory of
  Computing, to appear}.

\bibitem[Jin et~al., 2019b]{JLTX-19}
Jin, Y., Lu, P., Tang, Z.~G., and Xiao, T. (2019b).
\newblock Tight revenue gaps among simple mechanisms.
\newblock In {\em Proceedings of the Thirtieth Annual ACM-SIAM Symposium on
  Discrete Algorithms}, pages 209--228. SIAM.

\bibitem[Laffont and Robert, 1996]{LR-96}
Laffont, J.-J. and Robert, J. (1996).
\newblock Optimal auction with financially constrained buyers.
\newblock {\em Economics Letters}, 52(2):181--186.

\bibitem[Maskin and Riley, 1984]{MR-84}
Maskin, E. and Riley, J. (1984).
\newblock Optimal auctions with risk averse buyers.
\newblock {\em Econometrica}, pages 1473--1518.

\bibitem[Matthews, 1983]{mat-83}
Matthews, S.~A. (1983).
\newblock Selling to risk averse buyers with unobservable tastes.
\newblock {\em Journal of Economic Theory}, 30(2):370--400.

\bibitem[Myerson, 1981]{mye-81}
Myerson, R. (1981).
\newblock Optimal auction design.
\newblock {\em Mathematics of Operations Research}, 6:58--73.

\bibitem[Pai and Vohra, 2014]{PV-14}
Pai, M.~M. and Vohra, R. (2014).
\newblock Optimal auctions with financially constrained buyers.
\newblock {\em Journal of Economic Theory}, 150:383--425.

\bibitem[Richter, 2016]{Richter-18}
Richter, M. (2016).
\newblock Continuum mechanism design with budget constraints.
\newblock {\em Games and Economic Behavior, to appear}.

\bibitem[Riley and Samuelson, 1981]{RS-81}
Riley, J. and Samuelson, W. (1981).
\newblock Optimal auctions.
\newblock {\em American Economic Review}, 71:381--92.

\bibitem[Yan, 2011]{yan-11}
Yan, Q. (2011).
\newblock Mechanism design via correlation gap.
\newblock In {\em Proc. 22nd ACM Symp. on Discrete Algorithms}, pages 710--719.
  SIAM.

\end{thebibliography}

\appendix

\section*{Appendix}
%\setcounter{section}{0}
%{\renewcommand\thesection{\Alph{section}}
\section{Omitted Proof of \Cref{thm:public e bound imporved}}
\label{apx:public budget regular}
\begin{numberedtheorem}
{\ref{thm:public e bound imporved}}
An agent with public-budget utility and regular valuation distribution
is $(1, 1)$-close for price posting.  
\end{numberedtheorem}
% To prove \Cref{thm:public e bound imporved},
% it is sufficient to show for any quantile 
% $\exquant\in[0, 1]$, 
% the $\exquant$ ex ante optimal mechanism
% is a price-posting mechanism,
% i.e., $\revcurve(\exquant) = \cumprice(\exquant)$.
% To show this, 
% we write the ex ante optimal mechanism 
% as an optimization program,
% and 
% apply Lagrangian relaxation on
% the budget constraint.
% This leads to a new optimization
% program similar to an agent with linear 
% utility but with a 
% Lagrangian objective function.
% Following the technique 
% that price-posting revenue curve
% indicates the ex ante
% optimal mechanism
% for a linear agent,
% we 
% consider 
% the \emph{Lagrangian price-posting 
% revenue curve} 
% which characterizes the ex ante
% optimal mechanism for the 
% Lagrangian objective function.
% See further discussion about this 
% technique in \citet{AFHH-13}
% and \citet{FH-18}.
\begin{proof}
For an agent with public budget $\wealth$,
the $\exquant$ ex ante optimal mechanism 
is the solution of the following program,
\begin{align}
\label{eq:prog 1}
\begin{array}{ll}
\max\limits_{(\alloc,\price)} 
&\expect[\val]{\price(\val)}
\\
s.t.&(\alloc,\price) \text{ are IC, IR},
\\
&\expect[\val]{\alloc(\val)} = \exquant,
\\
&\price(\hval) \leq \wealth.
\end{array}
\end{align}
where $\hval$ is the highest possible
value of the agent.
Consider the Lagrangian
relaxation of the budget constraint 
in \eqref{eq:prog 1},
\begin{align}
\label{eq:prog 2}
\begin{array}{ll}
\min\limits_{\lagrange \geq 0}
\max\limits_{(\alloc,\price)} 
&\expect[\val]{\price(\val)}
+
\lagrange \wealth - \lagrange \price(\hval)
\\
s.t.&(\alloc,\price) \text{ are IC, IR},
\\
&\expect[\val]{\alloc(\val)} = \exquant.
\end{array}
\end{align}
Let $\lagrange^*$ be
the optimal solution in 
program~\eqref{eq:prog 2}. 
If we fix $\lagrange = \lagrange^*$ in 
program \eqref{eq:prog 2},
its inner maximization program
can be thought as a $\exquant$
ex ante optimal mechanism design
for a linear agent with 
Lagrangian objective function 
$\expect[\val]{\price(v)} - 
\lagrange^* \price(\hval)$.
Thus, we  
define the Lagrangian price-posting 
revenue curve $\lcumprice(\cdot)$
where $\lcumprice(\quant)$
is the maximum 
value of 
the Lagrangian objective $\expect[\val]{\price(\val)}
-\lagrange^* \price(\hval)
$
in price-posting mechanism with 
per-unit price $V(\quant)$.
For any
$\quant \in (0, 1]$, by the definition,
$\lcumprice(\quant) 
= \quant V(\quant) - \lagrange^* V(\quant)$.
For $\quant = 0$, notice that
the agent with 
$\hval$ is indifferent
between purchasing or not purchasing.
Thus,
by the definition,
$\lcumprice(\quant) = 0$ if $\quant = 0$.

Now, we consider the concave hull of 
the Lagrangian price-posting revenue curve 
$\lcumprice(\cdot)$ which 
we denote as $\clcumprice(\cdot)$.
Let $\quant\primed$ be the smallest
solution 
of equation 
$\lcumprice(\quant) = 
\quant\lcumprice'(\quant)$.
Since $\lcumprice(0) \leq 0$, $\lcumprice(1) = 0$
and $\lcumprice(\cdot)$ is continuous,
$\quant\primed$ always exists.
Then, for any $\quant \leq \quant\primed$,
$\clcumprice(\quant) =
\quant\,\lcumprice'(\quant\primed)$.
For any $\quant \geq \quant\primed$,
we show $\clcumprice(\quant) = \lcumprice(\quant)$
by the following arguments.
First notice that
% $\quant\lcumprice'(\quant) \geq 0$ 
% for all $\quant$,
$\lcumprice(\quant\primed) \geq 0$,
%which implies 
and hence $\quant\primed \geq \lambda^*$.
Consider $\lcumprice''(\quant) = V''(\quant)(\quant - \lambda^*) + 2V'(\quant)$.
Clearly, $V'(\quant) \leq 0$.
If $V''(\quant) \leq 0$, then 
$\lcumprice''(\quant) \leq 0$.
If $V''(\quant) > 0$, 
then 
$\lcumprice''(\quant) = V''(\quant)(\quant - \lambda^*) + 2V'(\quant) \leq \quant
V''(\quant) + 2V'(\quant) \leq 0$,
where $\quant V''(\quant) + 2V'(\quant)$ is 
non-positive due to the regularity of the 
valuation distribution.

To summarize, $\clcumprice(\cdot)$,
the concave hull of the Lagrangian price-posting revenue
curve
satisfies
\begin{align*}
    \clcumprice(\quant) = 
    \left\{
    \begin{array}{ll}
    \quant\,\lcumprice'(\quant\primed)
    & \text{ if } \quant \in 
    [0, \quant\primed] \\
    \lcumprice(\quant) & \text{ if } \quant \in 
    [\quant\primed, 1]
    \end{array}
    \right.
\end{align*}
Therefore, use the similar ironing 
technique based on the revenue curves
for linear agents with irregular valuation
distribution \citep[e.g.][]{mye-81, BR-89, AFHH-13},
\Cref{lem:lagrangian relaxation} (stated below)
suggests that 
the $\exquant$ ex ante optimal mechanism 
irons quantiles between $[0, \quant\primed]$
under $\exquant$ ex ante constraint, 
which is still a posted-pricing mechanism.
\end{proof}

\begin{lemma}[\citealp{AFHH-13}]\label{lem:lagrangian relaxation}
For incentive compatible and 
individual rational mechanism 
$(\alloc(\cdot), \price(\cdot))$ 
and an agent with any
Lagrangian price-posting revenue curve $\lcumprice(\quant)$,
the expected Lagrangian objective of the agent 
is upper-bounded by her expected marginal
Lagrangian objective
of the same allocation rule, i.e.,
\begin{align*}
\expect[\val]{\price(\val)} + \lagrange^*\price(\hval)
\leq \expect[\quant]{\clcumprice'(\quant)\cdot \alloc(V(\quant))}.
\end{align*} 
Furthermore, this inequality holds with equality if the allocation
rule $\alloc(\cdot)$ is constant all intervals of values $V(\quant)$
where $\clcumprice(\quant) > \lcumprice(\quant)$.
\end{lemma}

\section{Motivation of Assumptions for Private Budget}
\label{apx:private negative example}
In this section, we use 
two examples to demonstrate our assumption 
in \Cref{sec:private}
are reasonable.
In both examples where either the independence between
values and budgets
or
the constant quantile guarantee for 
the expected budget fails, while 
other assumptions
hold, 
the anonymous pricing is not a constant approximation to 
the optimal revenue.

%\begin{numberedassumption}{\ref{asp:value budget independent}}
% \begin{assumption}
% For each agent, her private value $\val$
% and private budget $\wealth$
% are drawn independently from 
% the valuation distribution
% $F$
% and 
% the budget distribution 
% $G$ 
% respectively.
% \end{assumption}
% %\end{numberedassumption}

% %\begin{numberedassumption}{\ref{asp:budget MHR}}
% \begin{assumption}
% For each agent,
% the budget distribution $G$ has 
% \emph{monotone hazard rate (MHR)}, i.e,
% $\frac{g(\wealth)}{1-G(\wealth)}$
% is monotonically non-decreasing 
% with respect to $\wealth$. 
% \end{assumption}
% %\end{numberedassumption}

\begin{example}[necessity of the independence between the value and budget distributions]\label{exp:independent fail}
Fix a large constant $h$.
Consider a single agent with value $\val$
drawn from $[1, h]$ 
with density function
$\frac{h}{h - 1}\frac{1}{\val^2}$,
and budget $\wealth = 2h - v$,
i.e., her value and budget 
are fully correlated.
A mechanism which
charges the agent $\val -2\epsilon$ with probability 
$1 - \frac{\epsilon}{h}$,
or $\wealth$ with probability $\frac{\epsilon}{h}$
for sufficient small positive $\epsilon$
is incentive compatible and has revenue $O(\ln h)$.
% However, the revenue of the anonymous pricing 
% approaches 1 when $h$ goes to infinity.
However, the revenue of 
the anonymous pricing is $O(1)$.
\end{example}

\begin{example}[necessity of the constant quantile of the expected budget]\label{exp:MHR fail}
Fix a large constant $\kappa$.
Consider the case with $n = \sqrt{\kappa}$
agents. 
For each player $i \in [n]$, 
set her valuation $\vali = i$ with probability 1,
and wealth $\wealth_i = i$ with probability $\frac{1}{i^2}$ 
and $\wealth_i = 0$ 
with probability $1-\frac{1}{i^2}$. 
For each agent, the probability 
her budget exceeds its expectation is at least $\sfrac{1}{\kappa}$. 
%By \citet{CHMS-10}, 
%\todo{need the correct citation}
The revenue of sequential posted-pricing,
\citep[cf.][]{CHMS-10},
is $O(\ln \kappa)$. 
However, the revenue of 
the anonymous pricing is $O(1)$. 
%Therefore, the approximation ratio is $n$. 
\end{example}

In \Cref{exp:MHR fail},
the price-posting revenue curve 
for agent $i$ is 
$\cumprice_i(\quant) = \quant\cdot i$
if $\quant \leq \frac{1}{i^2}$ 
and 
$\cumprice_i(\quant) = 0$ otherwise,
which is non-concave, and 
the ratio between 
the ex ante revenue and 
the price-posting revenue is infinity 
after quantile $\frac{1}{i^2}$. 
Moreover, \Cref{exp:MHR fail} indicates that 
under the assumption that for each agent, the probability 
her budget exceeds its expectation is at least $\sfrac{1}{\kappa}$, 
the lower bound on the approximation ratio of 
anonymous pricing to the 
optimal revenue is $\Omega(\ln \kappa)$. 
% Recall that for linear utility, 
% the regular assumption on valuation distribution
% required in \citet{AHNPY-18} is equivalent to
% assume the concavity of price-posting revenue curve.
% Therefore, for private-budget utility,
% we introduce the MHR assumption
% on budget distribution
% to guarantee the concave price-posting revenue curve.
%while the ex ante revenue curve
%is $\revcurve_i(\quant) = q\cdot h^i$
%if $\quant \leq h^{-i}$ 
%and 
%$\revcurve_i(\quant) = \frac{1 - \quant}{1 - h^{-i}}$
%otherwise.

%\input{Paper/appendix/appendix-allocation-payment_function}
%\input{Paper/appendix/appendix-concave_price-posting}
\section{Omitted Proofs in
\Cref{sec:private}}
\label{apx:private}

In this section, we give the omitted 
proofs for 
\Cref{lem:pricing function characterization}
and
\Cref{lem:private budget concave price-posting revenue curve}
%and 
%\Cref{thm: private budget MHR}
in \Cref{sec:private}.

\paragraph{Omitted Proofs
for \Cref{lem:pricing function characterization}}
\begin{numberedlemma}
{\ref{lem:pricing function characterization}}
For a single 
% item environment 
% with a single 
agent 
with private-budget utility,
in any incentive compatible mechanism, %$\mech$,
for all types with any fixed budget,
the mechanism %$\mech$ 
provides 
a convex and non-decreasing allocation-payment function,
%which maps the allocation to the payment, 
and subject to
this allocation-payment function,
each type will purchase as much as she wants 
until the budget constraint binds, 
or the unit-demand constraint binds, or 
the value binds (i.e., her marginal utility
becomes zero).
\end{numberedlemma}

\begin{proof}

\citet{mye-81} show that any
mechanisms $(\alloc, \price)$ for a single agent with linear utility
is incentive compatible 
(the agent does not prefer to misreport her value)
if and only if
a) $\alloc(\val)$ is non-decreasing;
b) $\price(\val) = \val\alloc(\val) 
-
\int_0^\val \alloc(t)dt
$.
Thus, 
given any non-decreasing allocation $\alloc$,
the payment $\price$ 
is uniquely pined down by the incentive constraints.

Comparing with the linear utility,
the incentive compatibility 
in the private-budget utility
guarantees that the agent does not prefer
to misreport 
either her value or budget.
If we relax the incentive constraints
such that she is only allowed to misreport her value,
Myerson result
already shows that 
for any fixed  budget level $\wealth$,
the allocation $\alloc(\val, \wealth)$ is 
non-decreasing in $\val$ and 
the payment
$\price(\val,\wealth)  = \val\alloc(\val, \wealth) 
-
\int_0^\val \alloc(t, \wealth)dt$
is uniquely pined down.
We define the allocation-payment function
$\APF_\wealth(\hat\alloc) = \max\{
\price(\val, \wealth)
+
\val\cdot(\hat\alloc - \alloc(\val,\wealth)):
\alloc(\val,\wealth) \leq \hat\alloc
\}
$ 
if $\hat\alloc \leq \alloc(\hval, \wealth)$;
and $\infty$ otherwise.
Given the characterization of
allocation and payment above,
this allocation-payment function
is well-defined, non-decreasing and convex.
\end{proof}

\begin{remark}
Unlike Myerson's result which give 
a sufficient and necessary condition
for
incentive compatible mechanisms
for 
agents with linear utility,
\Cref{lem:pricing function characterization}
only characterizes 
a necessary condition for private-budget
utility.
This condition
is already   enough for our arguments 
in \Cref{sec:private}.
\end{remark}

\paragraph{Omitted Proofs
for \Cref{lem:private budget concave price-posting revenue curve}}

\begin{numberedlemma}
{\ref{lem:private budget concave price-posting revenue curve}}
A single
agent with 
private-budget utility 
has a concave price-posting 
revenue curve $\cumprice$ 
if 
her value and budget 
are independently distributed,
the valuation distribution is regular,
and the budget distribution is MHR.
\end{numberedlemma}

\begin{proof}
Fixing any unit price $\price \in [0,\hval]$, 
the ex ante allocation $\quant$ and 
the expected price-posting revenue $\cumprice$ 
for posting price $\price$ are
\begin{align*}
  %  \begin{split}
        \quant &= 
        \frac{1}{\price}
        \left(
        \left(1 - G(\price)\right)\price 
        + \int_0^\price \wealth dG(\wealth)
        \right)
        (1-F(\price)) \\
        \cumprice &= 
        \left((1 - G(\price))\price + 
        \int_0^\price \wealth dG(\wealth)
        \right)
        (1-F(\price)).
%    \end{split}
%\end{align*}
\intertext{For notation simplicity, 
we let 
$\Delta(\price) = (1 - G(\price))\price + 
\int_0^\price \wealth dG(\wealth)$.
To check the concavity of the price-posting revenue curve, 
it is sufficient to show that
$\frac{\partial \cumprice}
{\partial \quant\partial \quant} \leq 0$.
%We first compute 
%$\frac{\partial \quant}{\partial \price}$ and 
%$\frac{\partial \cumprice}{\partial \price}$ 
%respectively,
}
%\begin{align*}
%    \begin{split}
        \frac{\partial \quant}{\partial \price} &= 
        -\frac{1}{\price}\Delta(\price)f(\price) + 
        \frac{1}{\price}(1-G(\price))(1-F(\price))
        -\frac{1}{\price^2}
        \Delta(\price)(1-F(\price)),\\
        \frac{\partial \cumprice}{\partial \price} 
        &= (1 - G(\price))(1 - F(\price)) 
        - \Delta(\price)f(\price). \\
 %   \end{split}
%\end{align*}
\intertext{The assumption 
that the valuation distribution $F$
has positive density on $[0, \hval]$ 
implies that
$\frac{\partial \quant}{\partial \price} < 0$.
Thus,}
%By our assumption on $F$ and $G$, 
%\begin{align*}
    \frac{\partial \cumprice}{\partial \quant} &= 
    \frac{\partial \cumprice}{\partial \price} \cdot \frac{\partial \price}{\partial \quant} 
    = 
    \frac{\price^2(1-G(\price))(1-F(\price)) 
    - \price^2\Delta(\price)f(\price)}
    {\price(1-G(\price))(1-F(\price))
    - \price\Delta(\price)f(\price) - 
    \Delta(\price)(1-F(\price))}.
%\end{align*}
\intertext{
Since $\frac{\partial \quant}{\partial \price} < 0$,
to show 
$\frac{\partial \cumprice}
{\partial \quant\partial \quant} \leq 0$,
it is sufficient to show $\frac{\partial}{\partial \price}
\frac{\partial \cumprice}{\partial \quant} \geq 0$.
Let $\psi(\price)$ and $\phi(\price)$ be 
the numerator and denominator 
of $\frac{\partial \cumprice}{\partial \quant}$, 
then
$\frac{\partial}{\partial \price}
\frac{\partial \cumprice}{\partial \quant} \geq 0$
is equivalent to 
$\psi'(\price)\phi(\price) 
- \psi(\price)\phi'(\price) \geq 0$.}
%\begin{align*}
%    \begin{split}
        \psi'(\price) &= 
        2\price(1-G(\price))(1-F(\price)) -
        \price^2g(\price)(1-F(\price)) - 
        \price^2(1-G(\price))f(\price) \\
        &\quad-2\price\Delta(\price)f(\price) -
        \price^2(1-G(\price))f(\price) -
        \price^2\Delta(\price)f'(\price),
        \\
        \phi'(\price) &= 
        (1-G(\price))(1-F(\price)) - 
        \price g(\price)(1-F(\price)) - 
        \price(1-G(\price))f(\price) \\
        &\quad-\Delta(\price)f(\price) - 
        \price(1-G(\price))f(\price) - 
        \price\Delta(\price)f'(\price) \\
        &\quad-(1-G(\price))(1-F(\price)) +
        \Delta(\price)f(\price).\\
%    \end{split}
%\end{align*}
%Then
\intertext{Thus, $\frac{1}{p}\left(\psi'(p)\phi(p) - \psi(p)\phi'(p)\right)$ equals}
%\begin{align*}
    %\begin{split}
        %&\quad\frac{1}{p}\left(\psi'(p)\phi(p) - \psi(p)\phi'(p)\right)\\
        %=&\quad
        %2(1-F(p))^2(1-G(p))^2p \\
        %&-
        %\Delta(p)(1-F(p))[2(1-F(p))(1-G(p)) + 
        %2f(p)(1-G(p))p - (1-F(p))g(p)p]\\
        %&+
        %\Delta^2(p)[2f(p)(1-F(p))+2f^2(p)p+(1-F(p))f'(p)p
        %]\\
        &2\Delta(\price)f(\price)(1-F(\price))
        (\Delta(\price)-(1-G(\price))\price) \\
        &+\Delta^2(\price)\price
        (2f^2(\price)+
        (1-F(\price))f'(\price)) \\
        &+(1-F(\price))^2(
        \Delta(\price)g(\price)\price-
        2(1-G(\price))
        (\Delta(\price)-(1-G(\price))\price))
    %\end{split}
%\end{align*}
\intertext{where 
%$2\Delta(\price)f(\price)(1-F(\price))
%(\Delta(\price)-(1-G(\price))\price) \geq 0$ 
%always holds, 
the first term is always non-negative 
and
the second term is non-negative
%$\Delta^2(\price)(2f^2(\price)\price+
%(1-F(\price))f'(\price)\price) \geq 0$ 
%holds 
since $F$ is a regular distribution.
%which implies that 
%$(1-F(\price))f'(\price) \geq -2f^2(\price)$
For the last term, 
when $g(p) = 0$, 
by the MHR assumption for the budget distribution, 
$G(p) = 0$ or $G(p) = 1$. 
In both cases, the last term is $0$. 
Hence, we only need to check that 
when $g(p) \neq 0$, 
%$$\Delta(p)g(p)p-2(1-G(p))(\Delta(p)-(1-G(p))p)\geq 0$$
the last term
is equivalent to }
&\Delta(\price)\price-
2\frac{1-G(\price)}{g(\price)}(\Delta(\price) 
- (1-G(\price))\price)\geq 0
\end{align*}
%\intertext{
Notice that the left hand side is zero 
when $\price = \lbudget$, 
and 
some non-negative value
when $p = \hbudget$. 
Therefore,
it is sufficient to show 
that the left hand side is monotone 
non-decreasing. 
Taking the derivative, it become
%}
%\begin{align*}
%    \begin{split}
%&(1-G(p))p + \Delta(p)
%-2(1-G(p))p 
%+2(1+\frac{1-G(p)}{g^2(p)}g'(p))(\Delta(p) - 
%(1-G(p))p)\\
%=&
  $\left(3+2(1-G(\price))\frac{g'(\price)}{g^2(\price)}
  \right)(\Delta(\price) - (1-G(\price))\price)$,
  which is non-negative since the budget distribution
$G$ is MHR.
%  }
%    \end{split}
%\end{align*}
\end{proof}

\section{Tightness of Reduction Framework}
\label{apx:tight framework}

\begin{lemma}
For any $\alpha, \beta \geq 1$ and $\eta = \alpha$, 
there exists $\{\cumprice_i\}_{i=1}^n$ 
and $\{\revcurve_i\}_{i=1}^n$ 
such that 
for each agent $i$, 
her price-posting revenue curve $\cumprice_i$ 
is $(\alpha, \beta)$-close for price posting
to her ex ante revenue curve $\revcurve_i$, 
the optimal price posting revenue 
for each agent
is an $\eta$-approximation to the  
optimal single agent revenue, 
and 
anonymous pricing on the price posting revenue curve  
is at most an $\left(\frac{e}{2(e-1)} \cdot \sqrt{\alpha\beta\eta}\right)$-approximation
to anonymous pricing on the ex ante revenue curve,
i.e.,
$
\frac{e}{2(e-1)} \cdot \sqrt{\alpha\beta\eta}
\cdot \pricerev(\{\cumprice_i\}_{i=1}^n)
\leq 
\pricerev(\{\revcurve_i\}_{i=1}^n)$. 
\end{lemma}
\begin{proof}
\begin{figure}[t]
\centering
\begin{tikzpicture}[scale = 0.6]

\draw (-0.2,0) -- (12.5, 0);
\draw (0, -0.2) -- (0, 6.4);

\begin{scope}[very thick]
\draw (0,0) -- (3, 6);
\draw (3,6) -- (12, 6);
\draw (0,0) -- (1, 1);
\draw (1,1) -- (12, 4);
\end{scope}

\draw (0, -0.8) node {$0$};

\draw (-0.5, 1) node {$1$};
\draw [dotted] (0, 1) -- (1, 1);
\draw [dotted] (1, 0) -- (1, 1);
\draw (1, -0.8) node {$\sfrac{1}{\beta}$};

\draw (-1, 6) node {$\alpha\sqrt{\beta}$};
\draw [dotted] (0, 6) -- (3, 6);
\draw [dotted] (3, 0) -- (3, 6);
\draw (3, -0.8) node {$\sfrac{1}{\sqrt{\beta}}$};

\draw (-0.7, 4) node {$\sqrt{\beta}$};
\draw [dotted] (0, 4) -- (12, 4);
\draw [dotted] (12, 0) -- (12, 4);
\draw (12, -0.8) node {$1$};

\draw (11, 5.5) node {$\revcurve_i$};
\draw (11, 3.2) node {$\cumprice_i$};

\end{tikzpicture}
\caption{\label{f:tight framework} 
The ex ante revenue curve 
and the price posting revenue curve 
for each agent $i$. 
}
\end{figure}
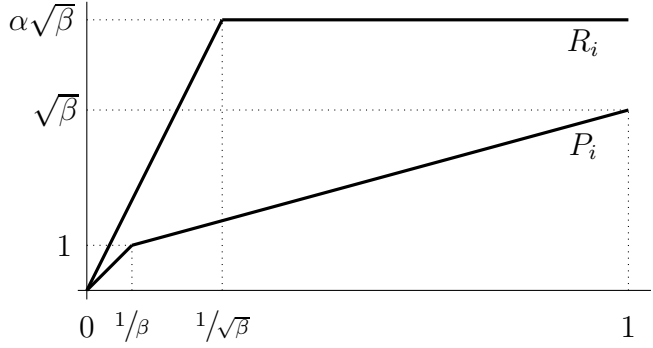

Consider an instance with $n = \sqrt{\beta}$ agents. 
For each agent $i$, 
her price posting revenue curve $\cumprice_i$ 
and 
her ex ante revenue curve $\revcurve_i$ 
are illustrated in \Cref{f:tight framework}.
% is 
% $(\alpha, \beta)$-close for price posting to 
% her ex ante revenue curve $\revcurve_i$, 
% and the 
% optimal price posting revenue 
% for each agent $i$
% is an $\eta$-approximation to the  
% optimal single agent revenue, 
% where $\eta = \alpha$. 

In this case, the optimal anonymous pricing on the ex ante revenue curve is 
to post effective price $\price = \alpha\beta$. 
The probability the item is sold is 
$1-(1-\sfrac{1}{\sqrt{\beta}})^{\sqrt{\beta}}
\geq \sfrac{e}{(e-1)}$.  
Therefore, $\pricerev(\{\revcurve_i\}_{i=1}^n) 
\geq \frac{e}{e-1} \cdot \alpha\beta$.

The anonymous pricing revenue on the price posting revenue curve 
is upper bounded by the optimal ex ante revenue 
on the price posting revenue curve. 
That is, 
$$\pricerev(\{\cumprice_i\}_{i=1}^n)
\leq \exanterev(\{\cumprice_i\}_{i=1}^n)
= \sum_{i=1}^n \cumprice_i(\frac{1}{\sqrt{\beta}}) 
\leq 2\sqrt{\beta}.$$
Therefore, 
$\pricerev(\{\revcurve_i\}_{i=1}^n) 
\geq \frac{e}{2(e-1)} \cdot \alpha\sqrt{\beta}
\cdot \pricerev(\{\cumprice_i\}_{i=1}^n) 
= \frac{e}{2(e-1)} \cdot \sqrt{\alpha\beta\eta}
\cdot \pricerev(\{\cumprice_i\}_{i=1}^n)$. 
\end{proof}

\end{document}